\documentclass[10pt, final, journal, twocolumn]{IEEEtran}
\usepackage{algorithm} 
\usepackage{algorithmic} 
\usepackage{multirow} 
\usepackage{amsmath}
\usepackage{xcolor}
\usepackage{stfloats}
\usepackage{amssymb}
\usepackage{amsmath}
\usepackage{mathrsfs}
\usepackage{cite}
\usepackage{url}
\usepackage{xcolor}
\usepackage{cite,graphicx,amsmath,amssymb}
\usepackage{subfigure}
\usepackage{citesort}
\usepackage{fancyhdr}
\usepackage{mdwmath}
\usepackage{mdwtab}
\usepackage{caption}
\usepackage{amsthm}
\usepackage{enumitem}
\usepackage{setspace}             
\usepackage{url}   
\usepackage{boondox-cal}
\usepackage{bm}
\usepackage{booktabs}
\usepackage{pifont}
\newcommand{\xmark}{\ding{55}}  
\DeclareMathAlphabet{\mathcal}{OMS}{cmsy}{m}{n}
\newtheorem{remark}{Remark}

\newtheorem{theorem}{Theorem}

\newtheorem{lemma}{Lemma}

\newtheorem{corollary}{Corollary}

\makeatletter
\def\ScaleIfNeeded{%
\ifdim\Gin@nat@width>\linewidth \linewidth \else \Gin@nat@width
\fi } \makeatother

\makeatletter
\renewcommand{\maketag@@@}[1]{\hbox{\m@th\normalsize\normalfont#1}}%
\makeatother

\begin{document}
\title{Joint Sensing and Covert Communications in RIS-NOMA Systems}
\author{
	Jiayi~Lei,~\IEEEmembership{Graduate Student Member,~IEEE,}
	Xidong~Mu,~\IEEEmembership{Member,~IEEE,}
	Tiankui~Zhang,~\IEEEmembership{Senior Member,~IEEE,}
	Wenjun~Xu,~\IEEEmembership{Senior Member,~IEEE,}
	Ping~Zhang~\IEEEmembership{Fellow,~IEEE}
	\thanks{ 
	This work was supported by the National Natural Science Foundation of
	China under Grant 62371068. \textit{(Corresponding authors: Tiankui Zhang; Wenjun Xu}.)
	
	Jiayi~Lei and Tiankui~Zhang are with the School of Information and Communication Engineering, Beijing University of Posts and Telecommunications, Beijing 100876, China (e-mail: \{leijiayi, zhangtiankui\}@bupt.edu.cn).
	
	Xidong Mu is with the Centre for Wireless Innovation (CWI), School of Electronics, Electrical Engineering and Computer Science, Queen's University Belfast, Belfast, BT3 9DT, U.K. (x.mu@qub.ac.uk).
	
	Wenjun Xu and Ping Zhang are with the State Key Laboratory
	of Networking and Switching Technology, Beijing University of Posts and
	Telecommunications, Beijing 100876, China (e-mail: \{wjxu,pzhang\}@bupt.edu.cn).
	}
}

\maketitle

\begin{abstract}
A reconfigurable intelligent surface (RIS)-assisted non-orthogonal multiple access (NOMA) system is investigated, where the transmitter (Alice) is a dual-functional radar-communication (DFRC) base station (BS) that aims to sense the location of a potential warden (Willie), while simultaneously transmitting public and covert signals to the legitimate users, Carol and Bob, respectively.
Both cases of known and unknown Willie locations are considered.
For the known-location case, assuming perfect channel state information (CSI) at Willie, a covert rate maximization is formulated with the joint optimization of active and passive beamforming, which is solved using successive convex approximation~(SCA), penalty method, and semidefinite relaxation~(SDR). 
For the unknown-location case, we propose to estimate Willie’s location via radar sensing and develop a sensing-based imperfect CSI model. In particular, the CSI error uncertainty is bounded by the sensing accuracy, which is characterized by the Cramér-Rao bound (CRB). Subsequently, a robust communication rate maximization problem is formulated under the constraints on quality-of-service~(QoS) of Carol, sensing accuracy, and covertness level. 
The Schur complement and S-procedure are employed to handle the non-convex constraints.	
Numerical results compare the system performance under the two cases, and demonstrate the significant covert performance superiority of the sensing-based imperfect CSI model and NOMA over the general norm-bounded imperfect CSI model and the orthogonal multiple access scheme.
Furthermore, the dual yet contradictory effects of sensing on covert communications are revealed. It is also found that Alice primarily utilizes Carol’s signal for sensing, while allocating almost all of Bob’s signal for communication.	
\end{abstract}

\section{Introduction}
With the coming of the Internet of Everything~(IoE) era, an increasing volume of sensitive, private, and confidential data is transmitted over open wireless channels, motivating extensive research on secure communications~\cite{7467419,9108996}. 
Cryptography and physical layer security~(PLS) techniques achieve information protection by exploiting message or protocol encryption and the inherent characteristics of wireless channels, respectively~\cite{9382022}. However, both approaches focus on preventing eavesdroppers from decoding the transmitted signals by concealing the content, and thus remain vulnerable to sophisticated adversaries equipped with advanced decryption algorithms or signal processing capabilities.
In contrast, covert communication aims to hide the very existence of the transmission itself from warden's detection, thereby providing a higher security level, without being cracked by technological progress at detectors~\cite{10090449}.
In general, the success of hiding the transmission behavior relies on the effective use of uncertainty or randomness to confuse and mislead warden, which can be introduced from channel fading, environment noise, power variation, location dynamic, etc~\cite{10278199,9849051,10341301}.
In~\cite{10278199}, the authors investigated the covert communication based on noise uncertainty, and proved that the maximum reliable covert rate against warden with optimal detector is $\mathcal{O}(\sqrt{N})$~bits per $N$ channel uses.
In~\cite{9849051}, user's covert rate was maximized via optimized transmit and jamming power.
In~\cite{10341301}, two unmanned aerial vehicles (UAVs) were
 deployed to introduce additional channel uncertainty through their high mobility, one acting as a relay assistant for Alice, and the other as an interference assistant targeting Willie.

Another effective approach to facilitate covert communications is to embed the covert signal within a public signal, which serves as the cover~\cite{8673809}. Non-orthogonal multiple access (NOMA) naturally realizes this mechanism by allowing multiple user signals to be superimposed and transmitted over the same time and frequency resources. This feature not only facilitates the concealment of covert signals but also enhances spectral efficiency and user fairness~\cite{9146329}.
It has been extensively verified that NOMA achieves higher performance gains when users exhibit markedly different channel conditions. 
Based on this, reconfigurable intelligent surfaces (RISs), which can reconfigure the electromagnetic environment through phase and amplitude controls, are considered to provide additional support for covert communications in NOMA systems. According to their operation on incident signals, RISs are generally classified into three types: reflective RISs, transmissive RISs, and simultaneously transmitting and reflecting RISs (STAR-RISs)~\cite{10550177}, suitable for scenarios where the transmitter and receiver are on the same side, on opposite sides, or on either side of RISs, respectively.
The contribution of RISs to covert communication can be summarized in two main aspects. On the one hand, by intelligently adjusting the phase shifts of their elements, RISs can regulate the channel gains of the public and covert users according to specific communication requirements, thereby further enhancing the performance of NOMA~\cite{9424177}. On the other hand, the phase shifts themselves introduce an inherent source of uncertainty, which helps confuse the warden and reduce its detection accuracy~\cite{9524501}.
Currently, the covert communication in RIS-NOMA systems have attracted increasing attention~\cite{9524501,10086651,10261424,10380743,10680088}.
The authors of \cite{9524501} investigated covert communications in a RIS assisted NOMA system, considering both uplink and downlink transmissions, and numerically demonstrated the necessity of  RIS and NOMA for achieving a positive covert rate.
In~\cite{10086651}, NOMA and rate splitting~(RS) were jointly applied to RIS aided users, and closed-form expressions for the detection error probability (DEP) and covert rate were derived.
The active RIS was employed in~\cite{10261424} to overcome the double path-loss attenuation in traditional RIS assisted NOMA covert communications. 
The authors of~\cite{10380743} and~\cite{10680088} further investigated STAR-RIS assisted covert communications in NOMA Systems. 
The defective successive interference cancellation~(SIC) was taken into account in~\cite{10380743}, where an effective covert rate  maximization problem was formulated to balance the reliability and covertness. 
In \cite{10680088}, the authors proposed an innovative STAR-RIS-assisted covert communication scheme in UAV-enabled NOMA systems, with convincing simulation results highlighting its effectiveness.

One crucial issue in covert communications is the acquisition of the warden's location and the corresponding channel state information~(CSI). Unlike legitimate users, the warden typically conceals its location, and due to its uncooperative nature, obtaining accurate CSI is extremely difficult. This poses a significant challenge for effective covert transmissions.
Most of the existing works neglect this practical issue, instead ideally assuming that the warden’s location is known and that perfect CSI is available.
Specially, the covert communications with an unknown-location warden were addressed in~\cite{9099441} and~\cite{9390411}.
The authors of \cite{9099441} introduced the concept of a covert threat region, defined as the set of potential warden locations at which the covert throughput fell below a specified threshold.
In~\cite{9390411}, a two-stage scheme was resourcefully proposed. In the first stage, the optimal warden location was determined based on covert outage probability, while in the second, transmit and jamming powers were optimized to maximize the covert transmission rate in the worst-case scenario. This approach effectively guarantees a lower bound for the system's covert communication capability.
Furthermore, the scenarios of imperfect CSI at the warden were addressed in \cite{10546987,9438645,10039715}. 
In~\cite{10546987}, the bounded CSI error model was adopted for the warden’s channel, and a robust average covert transmission rate maximization problem was formulated.
In \cite{9438645}, a RIS-assisted covert communication system was studied, where covert rate maximization was performed under three CSI scenarios for the warden’s link: perfect CSI, statistical imperfect CSI, and norm-bounded imperfect CSI.
In addition to the warden’s channel uncertainty, the imperfect CSI at the legitimate user was also  considered in~\cite{10039715}, and a robust transmission scheme was proposed based on continuous active beamforming and discrete passive beamforming.

However, the aforementioned works~\cite{9099441,9390411} do not explore methods for estimating the warden’s location. In addition, the studies in~\cite{10546987,9438645,10039715} adopt conventional imperfect channel models, which lack a clear theoretical basis for characterizing the uncertainty bound of CSI errors.
Inspired by the above observation, we propose to leverage the radar sensing capability of a dual-functional radar–communication (DFRC) base station (BS) to estimate the warden’s location~\cite{8869705}, and to utilize the sensing accuracy to characterize the uncertainty of imperfect CSI.
 In fact, there has emerged many research on covert communications in integrated sensing and communication~(ISAC) systems, but many of them involve certain other users as sensing targets, rather than the warden~\cite{10292914,10713261,10400197}.
In \cite{10974475,10659011,10497104}, the warden was explicitly regarded as a sensing target and the  imperfect CSI was accounted for the robust covert communications. Nevertheless, it is worth noting that these works still employ conventional CSI error models, whose uncertainty is agnostic to sensing mechanism.
Specifically, \cite{10974475} studied the joint secure and covert communications in STAR-RIS assisted ISAC systems, where both the warden and eavesdropper were set as sensing targets.
The case with uncertain target locations was considered, where the imperfect CSI followed the statistical CSI error model.
In \cite{10659011}, a additional dual-functional artificial noise (DFAN) was designed to induce uncertainty for covert communications and simultaneously sense multiple targets, including the warden. Three CSI error models for the warden’s channel were investigated, covering bounded-error, Gaussian-error and statistical-error models.
In~\cite{10497104}, the case of multiple wardens was examined, where the imperfect radar channels were considered as norm-bounded and probabilistic CSI error models.
Different from the works in \cite{10974475,10659011,10497104}, the study in \cite{10592837} adopted a sensing assisted imperfect CSI model that integrated the sensing capability into CSI uncertainty. However, it did not consider the crucial enhancements offered by RIS and NOMA for robust covert communications, and the relationship between sensing accuracy and robust covertness performance was not fully explored.

Motivated by the above, this paper investigates joint sensing and covert communications in  RIS-NOMA systems, where sensing is specifically employed to estimate the warden’s location, and an imperfect CSI model bounded by sensing accuracy is established. The main contributions are summarized as follows:
\begin{itemize}
	\item We present a joint sensing and covert communication framework that includes a DFRC-BS, a potential warden (Willie), and two legitimate indoor users: a public user (Carol), and a covert user (Bob). NOMA is employed to provide cover for the covert signal and to improve communication efficiency. Moreover, a transmissive RIS is deployed to establish reliable outdoor-to-indoor links and to further unleash the NOMA gain by intelligently reconfiguring the radio environment.
	\item We first consider an ideal scenario where Willie’s location is known, and only the communication functionality of the DFRC-BS is utilized. Under the worst-case condition with the minimum DEP~(MDEP), a covert rate maximization problem is formulated, subject to the quality-of-service (QoS) constraint of Carol. To solve this non-convex problem, the active and passive beamforming are jointly optimized by employing the successive convex approximation (SCA) technique, the penalty method, and semidefinite relaxation (SDR).
	\item We further investigate a more practical scenario in which Willie’s location is unknown. The radar sensing capability of the DFRC-BS is exploited to estimate Willie’s angle, based on which a sensing-based imperfect CSI model is proposed, with the CSI uncertainty bounded by the Cramér–Rao bound (CRB). A robust covert rate maximization problem is then formulated, subject to constraints on QoS and sensing accuracy. Building upon the algorithm developed for the known-location case, the problem is tackled with the aid of the Schur complement and the S-procedure.
	\item A series of numerical simulations are conducted, from which three key insights are obtained:
	1)~The perfect CSI model in the known-location scenario yields the highest covert communication rate and serves as an upper bound. The sensing-based imperfect CSI model significantly outperforms the general norm-bounded imperfect CSI model. Moreover, NOMA outperforms orthogonal multiple access (OMA).
	2)~Sensing plays a dual role in robust covert communication: an excessively high sensing accuracy requirement induces a trade-off between sensing and communication performance; otherwise, sensing positively contributes to enhancing robust covert transmission.
	3)~Sensing of Willie is primarily achieved through Carol’s signal, whereas Bob’s signal is mainly directed toward the RIS.
\end{itemize}

The comparison between this paper and the existing works is summarized as~\textbf{Table~\ref{tab:related_works}}.  The rest of this paper is organized as follows. 
Section~\ref{model} presents the framework for joint sensing and covert communications in RIS-assisted NOMA systems. Sections~\ref{perfect} and~\ref{imperfect} investigate the covert rate maximization problem under scenarios with known and unknown Willie, respectively. Section~\ref{simulation} shows the numerical results,and Section~\ref{conclusion} concludes the paper.

\textsl{Notations:} Scalars, vectors, and matrices are represented by lowercase letters, bold lowercase letters, and bold uppercase letters, respectively.~$(\cdot)^\mathrm{H}$ and $(\cdot)^\mathrm{T}$ denote the conjugate transpose and transpose of a matrix or vector, respectively. $\mathbb{E}(x)$ represents the expectation of scalar $x$. $||\mathbf{x}||$ denotes the Euclidean norm of the vector $\mathbf{x}$, and $||\mathbf{X}||_F$, $||\mathbf{X}||_*$, $||\mathbf{X}||_2$ represent the Frobenius norm, nuclear norm, and spectral norm, respectively. $\mathrm{Tr}(\mathbf{X})$ and $\mathrm{Rank}(\mathbf{X})$ denote the trace and rank of the matrix $\mathbf{X}$. $\text{diag}(\mathbf{x})$ creates a diagonal matrix with the vector $\mathbf{x}$ as its diagonal, and $\text{Diag}(\mathbf{X})$ extracts the diagonal elements of the matrix $\mathbf{X}$ to form a diagonal matrix.

\begin{table*}[htbp]
	\caption{Comparison with the existing works}
	\label{tab:related_works}
	\centering
	\renewcommand{\arraystretch}{0.9}
	\begin{tabular}{cccccc}
		\toprule
		\small
		& Uncertainty solution & Cover  &Sensing target & CSI model  \\
		\midrule
		\cite{9524501,10086651,10261424,10380743,10680088}  & RIS phase shifts & Non-orthogonal public signal & \xmark & Perfect \\
		\cite{9099441}        & Power variation  & \xmark      &\xmark & Statistical\\
		\cite{10546987}      & Location dynamic         &\xmark   & \xmark    & Bounded \\
		\cite{9438645}       & RIS phase shifts  & \xmark   & \xmark   & Perfect\&Statistical\&Bounded  \\
		\cite{10713261} &RIS phase shifts &Non-orthogonal public signal &Others &Statistical \\
		\cite{10659011}     & Artificial noise &\xmark & Warden  &Statistical\&Bounded\&Gaussian\\
		\cite{10592837}  & Channel fading   & \xmark  & Warden & Sensing-based imperfect\\
		This work    & RIS phase shifts  & Non-orthogonal public signal & Warden &Perfect\&Sensing-based imperfect \\
		\bottomrule
	\end{tabular}
\end{table*}

\section{System Model} \label{model}

As shown in Fig.~\ref{Fig1}, the considered communication system involves two indoor single-antenna users, namely Bob and Carol, where Bob is a covert user and Carol a public user.
A transmissive RIS, consisting of a uniform planar array~(UPA) with  $M=M_1\times M_2$ elements, is deployed on the window to facilitate indoor-outdoor public/covert communications. Alice is a DFRC-BS equipped with $N_t$ transmitted antennas and $N_r$ received antennas. 
Willie is a warden with single antenna, aiming to detect whether Alice is transmitting signals to Bob. 

\begin{figure}[tbp]
	\centering
	\includegraphics[width=2.8in]{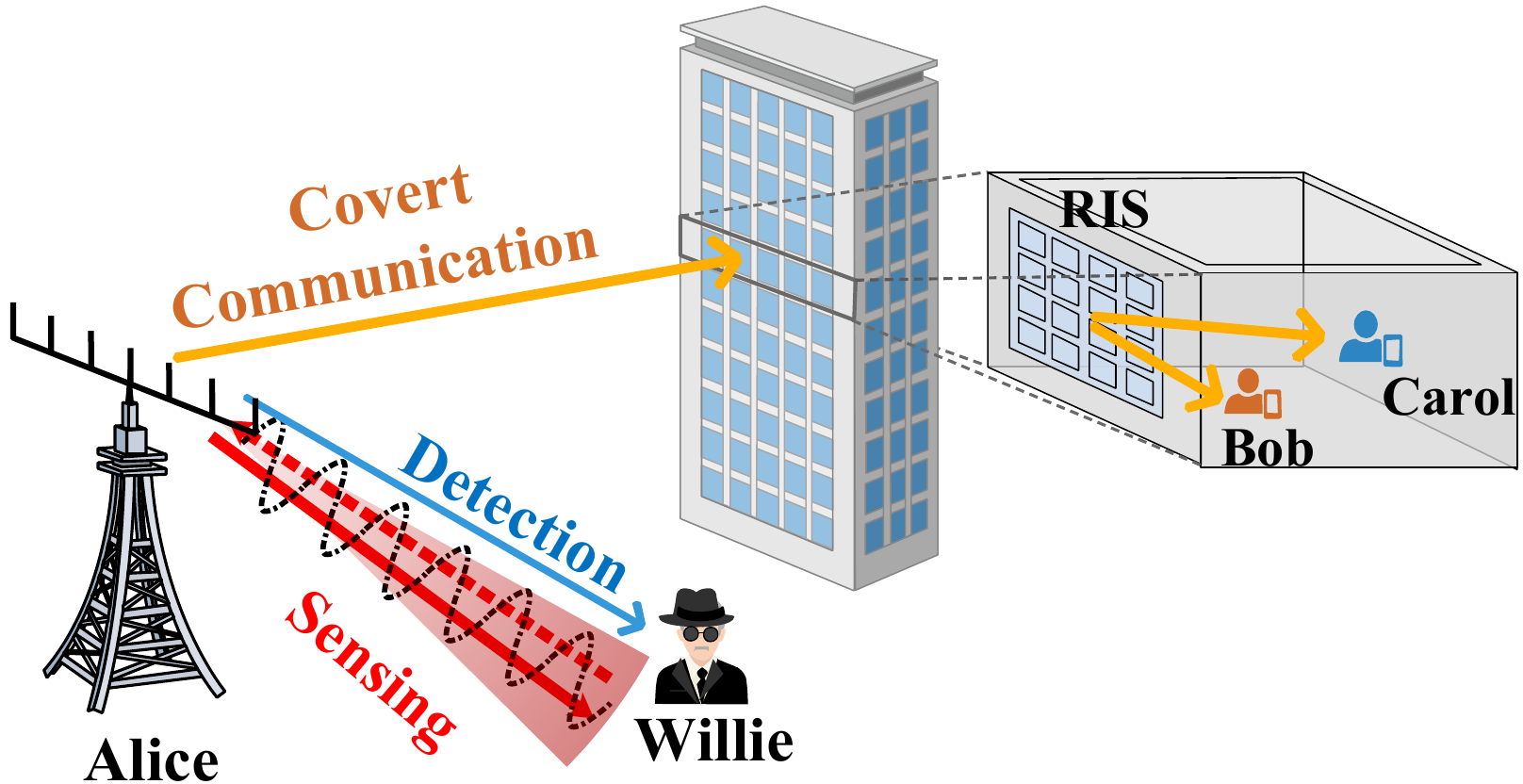}\\
	\caption{Illustration of joint sensing and covert communication in the RIS-NOMA system.\protect\footnotemark}\label{Fig1}
\end{figure}
\footnotetext{This work primarily focuses on the basic scenario with a single covert user and a single warden, where Willie is assumed to be static. Nevertheless, the proposed system framework can be extended to more complex scenarios, including multiple covert users or multiple wardens, or a dynamic Willie, by appropriately adapting the relevant problem formulation.}

\subsection{Channel Model}
The channels from Alice to RIS, from Alice to Willie and from RIS to Bob (Carol) are denoted by $\mathbf{H}_{\mathrm{AR}} \in \mathbb{C}^{M \times N_t}$, $\mathbf{h}_{\mathrm{AW}} \in \mathbb{C}^{N_t\times 1}$, and $\mathbf{h}_{\mathrm{RB}} (\mathbf{h}_{\mathrm{RC}}) \in \mathbb{C}^{M \times 1}$, respectively. The Alice-RIS link is assumed to experience the Rician fading, which is expressed as 
\begin{equation}\label{1H_AR}
	\mathbf{H}_{\mathrm{AR}} = \sqrt{\frac{\beta_\mathrm{R}}{d_{\mathrm{AR}}^{\alpha_\mathrm{R}}}} \left(\sqrt{\frac{\ell}{\ell+1}} \mathbf{H}_{\mathrm{AR}}^{\mathrm{LoS}} + \sqrt{\frac{1}{\ell+1}} \mathbf{H}_{\mathrm{AR}}^{\mathrm{NLoS}}\right),
\end{equation}
where $\beta_\mathrm{R}$ is the path loss at the reference distance, $d_{\mathrm{AR}}$ is the distance between Alice and RIS, $\alpha_\mathrm{R}$ is the path loss exponent, and $\ell$ is the Rician factor. The terms $\mathbf{H}_{\mathrm{AR}}^{\mathrm{LoS}}$ and $\mathbf{H}_{\mathrm{AR}}^{\mathrm{NLoS}}$ represent the line-of-sight (LoS) component and non-LoS (NLoS) component respectively. Specifically, the LoS component is given as 
\begin{equation}\label{H_AR}
	\mathbf{H}_{\mathrm{AR}}^{\mathrm{LoS}} = \mathbf{a}_\mathrm{R}(\gamma_\mathrm{A}, \phi_\mathrm{A})\mathbf{a}_\mathrm{A}(\theta_\mathrm{R})^\mathrm{H},
\end{equation}
where $\mathbf{a}_\mathrm{R}(\gamma_\mathrm{A}, \phi_\mathrm{A})$ is the receive steering vector at the RIS, with~$\gamma_\mathrm{A}$ and~$\phi_\mathrm{A}$ representing the elevation and azimuth angles of the incident signal from Alice, respectively.~$\mathbf{a}_\mathrm{A}$ is the transmit steering vector at Alice, with~$\theta_\mathrm{R}$ denoting the angle of departure~(AoD) from Alice to RIS. The NLoS component~$\mathbf{H}_{\mathrm{AR}}^{\mathrm{NLoS}}$ is assumed to follow a circularly
symmetric complex Gaussian distribution with zero mean and unit variance.

Given the proximity of the transmissive RIS to the indoor users, the channels from the RIS to Bob and Carol are assumed to be primarily LoS, and can thus be modeled as
\begin{equation}\label{h_RU}
	\mathbf{h}_{\mathrm{RB(RC)}} = \sqrt{\frac{\beta_{\mathrm{B(C)}}}{d_{\mathrm{RB(RC)}}^{\alpha_{\mathrm{B(C)}}}}} \, \mathbf{a}_\mathrm{R}(\gamma_{\mathrm{B(C)}}, \phi_{\mathrm{B(C)}}),
\end{equation}
where $\beta_{\mathrm{B(C)}}$,~$d_{\mathrm{RB(RC)}}$,~$\alpha_{\mathrm{B(C)}}$   correspond to the path loss at the reference distance, the distance between the RIS and Bob (or Carol), and the respective path loss exponent. The term~$\mathbf{a}_\mathrm{R}(\gamma_{\mathrm{B(C)}}, \phi_{\mathrm{B(C)}})$ refers to the RIS steering vector, with~$\gamma_{\mathrm{B(C)}}$ and~$\phi_{\mathrm{B(C)}}$ being the elevation and azimuth angles of departure from the RIS towards Bob (or Carol), respectively.

We consider a challenging scenario where Willie can freely select a location that optimizes detection by positioning itself near Alice with a direct LoS path. As such, the channel from Alice to Willie is also modeled as a LoS-dominated channel, expressed as 
\begin{equation}\label{h_AW}
	\mathbf{h}_{\mathrm{AW}} = \sqrt{\frac{\beta_\mathrm{W}}{d_{\mathrm{AW}}^{\alpha_\mathrm{W}}}}\mathbf{a}_\mathrm{A}(\theta_\mathrm{W}),
\end{equation} 
where $\beta_\mathrm{W}$ is the path loss at the reference distance,~$d_{\mathrm{AW}}$ denotes the distance between Alice and Willie, and~$\alpha_\mathrm{W}$ is the path loss exponent. The term~$\mathbf{a}_\mathrm{A}(\theta_\mathrm{W})$ represents the steering vector at Alice, with~$\theta_\mathrm{W}$ being the AoD from Alice to Willie.

The detailed expression for the steering vector of the UPA-based RIS, which appears in \eqref{H_AR} and \eqref{h_RU}, is given by
\begin{equation}
	\mathbf{a}_{\mathrm{R}}(\gamma,\phi) = \mathbf{a}_{\mathrm{h}}(\gamma,\phi) \otimes \mathbf{a}_{\mathrm{v}}(\gamma,\phi),
\end{equation}
where
\begin{equation*}
	\mathbf{a}_{\mathrm{h}}(\gamma,\phi) = \left[ 1,\; e^{j\frac{2\pi}{\lambda}d\sin\gamma\cos\phi},\; \dots,\; e^{j\frac{2\pi}{\lambda}d(M_1-1)\sin\gamma\cos\phi} \right]^\mathrm{T},
\end{equation*}
\begin{equation*}
	\mathbf{a}_{\mathrm{v}}(\gamma,\phi) = \left[ 1,\; e^{j\frac{2\pi}{\lambda}d\sin\gamma\sin\phi},\; \dots,\; e^{j\frac{2\pi}{\lambda}d(M_2-1)\sin\gamma\sin\phi} \right]^\mathrm{T},
\end{equation*}
and $\otimes$ denotes the Kronecker product. The steering vector of Alice, involved in \eqref{H_AR} and \eqref{h_AW}, is given by
\begin{equation*}
	\mathbf{a}_{\mathrm{A}}(\theta) = \left[ 1,\; e^{j\frac{2\pi}{\lambda}d\sin\theta},\; e^{j\frac{2\pi}{\lambda}2d\sin\theta},\; \dots,\; e^{j\frac{2\pi}{\lambda}(N_t-1)d\sin\theta} \right]^\mathrm{T},
\end{equation*}
where $d$ is the antenna or element spacing at Alice and RIS.

It is worth noting that, since Bob and Carol are legitimate users and the RIS is a manually deployed node assisting the communication, Alice is assumed to have accurate knowledge of their locations, which allows the exact CSIs in~\eqref{1H_AR} and~\eqref{h_RU} to be obtained. In contrast, as Willie is typically non-cooperative, its accurate location is difficult to acquire, which makes the perfect CSI in~\eqref{h_AW} challenging to obtain.
Let $\mathbf{\Lambda} = \operatorname{diag}(e^{j\theta_1}, e^{j\theta_2}, \dots, e^{j\theta_M})$ denote the passive beamforming matrix of RIS, where $\theta_m \in [0, 2\pi)$ represents the adjustable phase shift of the $m$-th transmitting element. Based on the aforementioned channel models, the equivalent channels from Alice to Bob and Carol via the RIS can be expressed as
$\mathbf{g}_\mathrm{B} = \mathbf{h}_{\mathrm{RB}}^\mathrm{H}\mathbf{\Lambda}\mathbf{H}_{\mathrm{AR}}$ and $\mathbf{g}_\mathrm{C} = \mathbf{h}_{\mathrm{RC}}^\mathrm{H}\mathbf{\Lambda}\mathbf{H}_{\mathrm{AR}}$, respectively.

\subsection{NOMA Transmission}
To provide a non-orthogonal cover signal against detection by Willie, and also to improve the spectral efficiency of the multi-user communication system, NOMA is employed to enable simultaneous public and covert transmissions.
Let $s_\mathrm{B}[l]$ and $s_\mathrm{C}[l]$ denote the covert and public symbols in the $l$-th channel use for Bob and Carol respectively, with $\mathbb{E}(|s_\mathrm{B}[l]|^2) = \mathbb{E}(|s_\mathrm{C}[l]|^2)= 1$. The superimposed signal transmitted at Alice is expressed as
\begin{equation}\label{signal}
	\mathbf{x}[l] = \mathbf{w}_\mathrm{B}s_\mathrm{B}[l]+\mathbf{w}_\mathrm{C}s_\mathrm{C}[l],
\end{equation} 
where $\mathbf{w}_\mathrm{B} \in \mathbb{C}^{N_t \times 1}$ and $\mathbf{w}_\mathrm{C} \in \mathbb{C}^{N_t \times 1}$ denote the precoding vector for Bob and Carol, respectively. Then, the signals received at Bob and Carol are given as 
\begin{equation}
	y_\mathrm{B}[l] = \mathbf{g}_\mathrm{B}\left(\mathbf{w}_\mathrm{B}s_\mathrm{B}[l]+\mathbf{w}_\mathrm{C}s_\mathrm{C}[l]\right)+n_\mathrm{B}[l],
\end{equation}
and
\begin{equation}
	y_\mathrm{C}[l] = \mathbf{g}_\mathrm{C}\left(\mathbf{w}_\mathrm{B}s_\mathrm{B}[l]+\mathbf{w}_\mathrm{C}s_\mathrm{C}[l]\right)+n_\mathrm{C}[l],
\end{equation}
where $n_\mathrm{B}[l]$ and $n_\mathrm{C}[l]$ denote the Additive White Gaussian Noise (AWGN) with the variance of $\sigma_\mathrm{B}^2$ and $\sigma_\mathrm{C}^2$. 

Without loss of generality, it is assumed that Carol is the weak user, i.e., the channel gain at Bob is higher than that at Carol, $\|\mathbf{g}_\mathrm{B}\| > \|\mathbf{g}_\mathrm{C} \|$. Therefore, a \textit{``Carol-to-Bob"} SIC order is set, meaning that Carol's signal is decoded first, followed by the subsequent extraction of Bob's desired signal, while Carol directly receives its signal by treating Bob’s signal as interference.
Such set does not affect the energy-ratio based detection at Willie, and help to achieve an effective cover for Bob, as more power is allocated to Carol than Bob.

Based on the \textit{``Carol-to-Bob"} SIC order, the signal-to-interference-plus-noise ratio (SINR) at Bob is given by:
\begin{equation}\label{sinrB}
\text{SINR}_\mathrm{B} = \frac{| \mathbf{g}_\mathrm{B} \mathbf{w}_\mathrm{B}|^2}{\sigma_\mathrm{B}^2}.	
\end{equation}
The SINR for decoding Carol's signal at Bob and Carol are expressed as
$\text{SINR}_{\mathrm{B} \to \mathrm{C}} = \frac{| \mathbf{g}_\mathrm{B} \mathbf{w}_\mathrm{C}|^2}{| \mathbf{g}_\mathrm{B} \mathbf{w}_\mathrm{B}|^2 + \sigma_\mathrm{B}^2}$ and
$\text{SINR}_{\mathrm{C} \to \mathrm{C}} = \frac{| \mathbf{g}_\mathrm{C} \mathbf{w}_\mathrm{C}|^2}{| \mathbf{g}_\mathrm{C} \mathbf{w}_\mathrm{B}|^2 + \sigma_\mathrm{C}^2},
$ respectively. To make the implementation of SIC feasible, the achievable SINR for Carol is given as~\cite{7555306}
\begin{equation}
	\text{SINR}_\mathrm{C} =  \min\Bigg\{ \frac{| \mathbf{g}_\mathrm{B} \mathbf{w}_\mathrm{C}|^2}{| \mathbf{g}_\mathrm{B} \mathbf{w}_\mathrm{B}|^2 + \sigma_\mathrm{B}^2},  \frac{| \mathbf{g}_\mathrm{C} \mathbf{w}_\mathrm{C}|^2}{| \mathbf{g}_\mathrm{C} \mathbf{w}_\mathrm{B}|^2 + \sigma_\mathrm{C}^2}\Bigg\}.
\end{equation}

\subsection{Detection at Willie}\label{C}
To detect the presence of covert communication, Willie is faced with a binary hypothesis test between: (i) the null hypothesis $\mathcal{H}_0$, where Alice transmits exclusively to Carol; and (ii) the alternative hypothesis $\mathcal{H}_1$, where Alice concurrently transmits to both Bob and Carol. Under these two hypotheses, the received signal in the $l$-th channel use at Willie is given as 
\begin{equation}\label{signal}
	y_W[l] = 
	\begin{cases}
		\mathbf{h}_\mathrm{AW}^\mathrm{H}\mathbf{w}_\mathrm{C}s_\mathrm{C}[l]+n_\mathrm{W}[l], & \mathcal{H}_0,  \\
		\mathbf{h}_\mathrm{AW}^\mathrm{H}\big(\mathbf{w}_\mathrm{C}s_\mathrm{C}[l]+\mathbf{w}_\mathrm{B}s_\mathrm{B}[l]\big)+n_\mathrm{W}[l], & \mathcal{H}_1,  \\
	\end{cases}
\end{equation}
where $n_\mathrm{W}[l]\sim{\cal CN} (0,\sigma_\mathrm{W}^2)$ denotes the AWGN at Willie.
Willie then needs to make a binary decision, denoted by $\mathcal{D}_0$ or $\mathcal{D}_1$, in favor of the hypothesis $\mathcal{H}_0$ or $\mathcal{H}_1$, respectively. In general, the detection performance of Willie is measured by the DEP, which is defined as 
\begin{equation}
	\label{DEP}
		\xi =\mathbb{P}_\mathrm{FA} + \mathbb{P}_\mathrm{MD},
\end{equation}
where $\mathbb{P}_\mathrm{FA} = \text{Pr}\big(\mathcal{D}_1|\mathcal{H}_0 \big)$ and $\mathbb{P}_\mathrm{MD} = \text{Pr}\big(\mathcal{D}_0|\mathcal{H}_1 \big)$
denote the false alarm probability and missed detection probability, respectively.

We consider the worst case where Willie follows the optimal decision rule to achieve the MDEP by applying the Neyman–Pearson criterion. Specifically, according to \eqref{signal}, the likelihood functions of $y_\mathrm{W}[l]$ under $\mathcal{H}_0$ and $\mathcal{H}_1$ can be expressed as 
\begin{equation}
\begin{cases}
	\mathbb{P}_0 = \frac{1}{\pi\lambda_0}e^{-\frac{|y_\mathrm{W}[l]|^2}{\lambda_0}}, & \mathcal{H}_0,  \\
    \mathbb{P}_1 = \frac{1}{\pi\lambda_1}e^{-\frac{|y_\mathrm{W}[l]|^2}{\lambda_1}}, & \mathcal{H}_1,  \\
\end{cases}	
\end{equation}
where $\lambda_0$ and $\lambda_1$ represent the mean powers for the received signals under $\mathcal{H}_0$ and $\mathcal{H}_1$, and are given by 
\begin{equation}\label{lam}
	\begin{cases}
		\lambda_0= |\mathbf{h}_\mathrm{AW}^\mathrm{H}\mathbf{w}_\mathrm{C}|^2+\sigma_\mathrm{W}^2, & \mathcal{H}_0,  \\
		\lambda_1= |\mathbf{h}_\mathrm{AW}^\mathrm{H}\mathbf{w}_\mathrm{B}|^2 +|\mathbf{h}_\mathrm{AW}^\mathrm{H}\mathbf{w}_\mathrm{C}|^2+\sigma_\mathrm{W}^2, & \mathcal{H}_1. 
	\end{cases}
\end{equation}
Then, the joint likelihood ratio over all $L$ symbols is used as the decision metric for detecting the presence of covert transmission, i.e., 
\begin{equation}
	\label{detection}
	 \frac{\mathbb{P}_0^L}{\mathbb{P}_1^L}
 \overset{\mathcal{D}_L}{\underset{\mathcal{D}_0}{\gtrless}} 1.
\end{equation}
Based on this detection rule, the corresponding MDEP can be derived as 
\begin{equation}
	\label{MDEP}
		\xi^*  = 1 - \mathcal{V}_T(\mathbb{P}_0^L, \mathbb{P}_1^L),
\end{equation}
where $\mathcal{V}_T(\mathbb{P}_0^L, \mathbb{P}_1^L)$ represents the total variation distance between $\mathbb{P}_0^L$ and $\mathbb{P}_1^L$.

\subsection{Covert Constraint}\label{D}
From the perspective of the Alice-to-Bob transmission, our objective is to ensure covertness by limiting Willie's achievable MDEP. The covert constraint is given as $\xi^* \geq 1- \epsilon$, where~$\epsilon \in [0,1]$ denotes the required covertness level. Nevertheless, the mathematical form of~\eqref{MDEP} is complicated and difficult to handle. To this end, we apply Pinsker's inequality to obtain a lower bound of the MDEP\cite{8714018}, which is given as\footnote{As  proven in \cite{8714018}, $\mathcal{D}(\mathbb{P}_0^L||\mathbb{P}_1^L)$ determines a tighter low bound on $\xi^*$ than  $\mathcal{D}(\mathbb{P}_1^L||\mathbb{P}_0^L)$ for Gaussian signalling. Therefore, we adopt the former KL divergence in our analysis.}
\begin{equation}
	\label{D1}
	\xi^* \ge 	1- \sqrt{\frac{1}{2}\mathcal{D}(\mathbb{P}_0^L, \mathbb{P}_1^L)},
\end{equation}
where $\mathcal{D}(\mathbb{P}_0^L, \mathbb{P}_1^L)$ is the Kullback-Leibler (KL) divergence from $\mathbb{P}_0^L$ to $\mathbb{P}_1^L$. Note that the received signal $y_\mathrm{W}[l]$ is identically independently distributed (i.i.d.), we have 
 \begin{equation}\label{D2}
 	\mathcal{D}(\mathbb{P}_0^L, \mathbb{P}_1^L)  = L \times \mathcal{D}(\mathbb{P}_0||\mathbb{P}_1),
 \end{equation}
and $\mathcal{D}(\mathbb{P}_0||\mathbb{P}_1)$ is given as
 \begin{equation}\label{D3}
 	\mathcal{D}(\mathbb{P}_0||\mathbb{P}_1) = \int_{y_\mathrm{W}[i]} \mathbb{P}_0\ln \frac{\mathbb{P}_0}{\mathbb{P}_1} \ d (y_\mathrm{W}[i]) = \ln \frac{\lambda_1}{\lambda_0}+ \frac{\lambda_0}{\lambda_1}-1,
 \end{equation}
 Based on \eqref{D1} $\sim$ \eqref{D3}, the covert constraint can be further expressed as 
\begin{equation}\label{ic}
	\ln \frac{\lambda_1}{\lambda_0}+ \frac{\lambda_0}{\lambda_1}-1 \le \frac{2\epsilon^2}{L}.
\end{equation}

To make \eqref{ic} more tractable, we reformulate it as follows. It is easy to know the function $ln(x)+\frac{1}{x}$ has the minimum value at $x=1$. Furthermore, the function is monotonically increasing for $x>1$ and decreasing for $0<x<1$.
Based on this property, the equation $ln(x)+\frac{1}{x}= \frac{2\epsilon^2}{L}+1$ admits two roots, denoted by $x_1(\epsilon)$ (with $0<x_1(\epsilon)<1$) and $x_2(\epsilon)$  (with $x_2(\epsilon)>1$) respectively. Therefore, the inequality $ln(x)+\frac{1}{x} \leq \frac{2\epsilon^2}{L}+1$ is equivalent to $x_1(\epsilon)<x<x_2(\epsilon)$. 
By treating $\frac{\lambda_1}{\lambda_0}$ as the variable $x$, and noting that $\lambda_1 \geq \lambda_0$ always holds according to \eqref{lam}, the constraint \eqref{ic} can thus be rewritten as $\frac{\lambda_1}{\lambda_0} \leq x_2(\epsilon)$.

\section{Covert Communication Under Known Warden Location}\label{perfect}

In this section, we consider an ideal scenario where Alice has full knowledge of Willie’s location, and the perfect channel vector $\mathbf{h}_{\mathrm{AW}}$ in~\eqref{h_AW} is assumed to be given and deterministic. Under this ideal condition, Alice operates solely in its communication role, without performing its sensing functionality. Our objective is to maximize Bob’s covert rate through the joint design of active beamforming at Alice and passive beamforming at RIS, while ensuring that the QoS requirement on Carol’s public communication is satisfied. To tackle the formulated non-concave problem, we propose a solution approach based on SCA technique, penalty method and SDR.

\subsection{Problem Formulation}
Aiming to maximize Bob's rate subject to the constraints on covertness and Carol's QoS, the optimization problem is given as
\begin{subequations}\label{P1}
	\begin{align}
		(\mathcal{P}1)	~~~& ~~\mathop {{\rm{Max}}}\limits_{\mathbf{w}_\mathrm{B},\mathbf{w}_\mathrm{C},\mathbf{\Lambda}} \log_2\big(1+\text{SINR}_\mathrm{B}\big), \label{1a} \\
		~~\text{s.t.} ~~~~
		& \log_2\big(1+\text{SINR}_\mathrm{C}\big) \geq R_0, \label{1b}\\
		& \| \mathbf{w}_\mathrm{B} \|^2 + \| \mathbf{w}_\mathrm{C} \|^2 \leq P, \label{c} \\
		& \theta_n \in [0, 2\pi), n=1,2,...,M,  \label{d}\\
		& \frac{\lambda_1}{\lambda_0} \leq x_2(\epsilon),  \label{e}
    \end{align}
\end{subequations}
where $R_0$ and $P$ denote the required minimum rate of Carol and maximum available power of Alice, respectively. Specifically, \eqref{1b} ensures that Carol’s public transmission rate is no less than $R_0$; \eqref{c} limits Alice's transmission power to the available budget $P$; \eqref{d} defines the feasible range for the RIS phase shifts; and \eqref{e} represents the covertness constraint against Willie’s detection.

To transform \eqref{P1} into a more tractable form, we reformulate the involved vector variables into matrix representations. Denoting the RIS beamforming vector as~$\mathbf{v}= [e^{j\theta_1},e^{j\theta_2},...,e^{j\theta_M}]^\mathrm{H}$, the channel vectors at Bob and Carol can be expressed as $\mathbf{g}_\mathrm{B}= \mathbf{v}^\mathrm{H}\mathbf{G_\mathrm{B}}$ and $\mathbf{g}_\mathrm{C}= \mathbf{v}^\mathrm{H}\mathbf{G_\mathrm{C}}$, where $\mathbf{G_\mathrm{B}}=\operatorname{diag}(\mathbf{h}_\mathrm{RB}^\mathrm{H})\mathbf{H}_\mathrm{AR}$ and 
$\mathbf{G_\mathrm{C}} =\operatorname{diag}(\mathbf{h}_\mathrm{RC}^\mathrm{H})\mathbf{H}_\mathrm{AR}$.
We further define the RIS beamforming matrix as $\mathbf{V} = \mathbf{v}\mathbf{v}^\mathrm{H}$, which satisfies $\mathbf{V} \succeq 0$, $\operatorname{Rank}(\mathbf{V}) = 1$, and $\operatorname{Diag}(\mathbf{V}) = \mathbf{1}_M$. Similarly, the transmit precoding matrices at Alice for Bob and Carol are defined as $\mathbf{W}_B = \mathbf{w}_\mathrm{B}\mathbf{w}_\mathrm{B}^\mathrm{H}$ and 
$\mathbf{W}_\mathrm{C} = \mathbf{w}_\mathrm{C}\mathbf{w}_\mathrm{C}^\mathrm{H}$ respectively, with the constraints $\mathbf{W}_\mathrm{B}\succeq 0$, $\mathbf{W}_\mathrm{C}\succeq 0$, and $\operatorname{Rank}(\mathbf{\mathbf{W}_\mathrm{B}})= \operatorname{Rank}(\mathbf{\mathbf{W}_\mathrm{C}}) = 1$. 
 Based on the above definitions, the equivalent channel gains at Bob and Carol can be rewritten as $|\mathbf{g}_\mathrm{B}\mathbf{w}_\mathrm{B}|^2 = \operatorname{Tr}(\mathbf{V}\mathbf{G}_\mathrm{B}\mathbf{W}_\mathrm{B}\mathbf{G}_\mathrm{B}^\mathrm{H})$,
$|\mathbf{g}_\mathrm{B}\mathbf{w}_\mathrm{C}|^2 = \operatorname{Tr}(\mathbf{V}\mathbf{G}_\mathrm{B}\mathbf{W}_\mathrm{C}\mathbf{G}_\mathrm{B}^\mathrm{H})$, respectively. 

Moreover, since $\log_2(1+x)$ is a monotonically increasing function, \eqref{1a} implies maximizing $\text{SINR}_\mathrm{B}$, which corresponds to maximizing the channel gain at Bob as per \eqref{sinrB}, and \eqref{1b} is equivalent to imposing a lower bound on $\text{SINR}_\mathrm{C}$. Therefore, The problem $\mathcal{P}1$ can be reformulated as 
 \begin{subequations}\label{P1.1}
	\begin{align}
    \hspace{-5pt} (\mathcal{P}1.1)&~~~\mathop {{\rm{Max}}}\limits_{\mathbf{W}_\mathrm{B},\mathbf{W}_\mathrm{C},\mathbf{V}} 
		\operatorname{Tr}(\mathbf{V}\mathbf{G}_\mathrm{B}\mathbf{W}_\mathrm{B}\mathbf{G}_\mathrm{B}^\mathrm{H}),\label{2a} \\
	&\hspace{-25pt}\text{s.t.} ~~
		\rho\operatorname{Tr}(\mathbf{V}\mathbf{G}_\mathrm{C}\mathbf{W}_\mathrm{B}\mathbf{G}_\mathrm{C}^\mathrm{H}) +\rho\sigma_\mathrm{C}^2 \leq \operatorname{Tr}(\mathbf{V}\mathbf{G}_\mathrm{C}\mathbf{W}_\mathrm{C}\mathbf{G}_\mathrm{C}^\mathrm{H}), \label{2b} \\
	&\hspace{-5pt}	\rho\operatorname{Tr}(\mathbf{V}\mathbf{G}_\mathrm{B}\mathbf{W}_\mathrm{B}\mathbf{G}_\mathrm{B}^\mathrm{H}) +\rho\sigma_\mathrm{B}^2 \leq  \operatorname{Tr}(\mathbf{V}\mathbf{G}_\mathrm{B}\mathbf{W}_\mathrm{C}\mathbf{G}_\mathrm{B}^\mathrm{H}), \label{2c} \\
	&\hspace{-5pt}\operatorname{Tr}(\mathbf{W}_\mathrm{B})+ \operatorname{Tr}(\mathbf{W}_\mathrm{C}) \leq P, \label{2d}\\
	&\hspace{-5pt} \mathbf{V} \succeq 0, \label{2e}\\
	&\hspace{-5pt}\operatorname{Rank}(\mathbf{V}) = 1, \label{2f}\\
	&\hspace{-5pt}\operatorname{Diag}(\mathbf{V}) = \mathbf{1}_M, \label{2g}\\
	&\hspace{-5pt}\mathbf{W}_\mathrm{B}\succeq 0, \mathbf{W}_\mathrm{C}\succeq 0, \label{2h}\\
	&\hspace{-5pt}\operatorname{Rank}(\mathbf{\mathbf{W}_\mathrm{B}}) = \operatorname{Rank}(\mathbf{\mathbf{W}_\mathrm{C}}) = 1, \label{2i}\\
	&\hspace{-5pt} \frac{\operatorname{Tr}(\mathbf{h}_\mathrm{AW}\mathbf{h}_\mathrm{AW}^\mathrm{H}(\mathbf{W}_\mathrm{C}+\mathbf{W}_\mathrm{B}))+\sigma_\mathrm{W}^2}{\operatorname{Tr}(\mathbf{h}_\mathrm{AW}\mathbf{h}_\mathrm{AW}^\mathrm{H}\mathbf{W}_\mathrm{C})+\sigma_\mathrm{W}^2} \leq x_2(\epsilon), \label{2j} 
	\end{align}
\end{subequations}
where $\rho = 2^{R_0} -1$ represents the minimum required SINR for Carol. The constraints \eqref{2b} and \eqref{2c} are equivalent to \eqref{1b}, while \eqref{2d}, \eqref{2g}, and \eqref{2j} correspond to \eqref{c}, \eqref{d}, and \eqref{e}, respectively.

\begin{figure*}[hbp]
	\hrulefill
	\begin{equation}\label{AB_le}
		\begin{aligned}
			{f}_\text{lb}\left(\mathbf{A},\mathbf{B},\mathbf{A}_0,\mathbf{B}_0\right)&=\frac{1}{2}\left\|\mathbf{A}_0+\mathbf{B}_0\right\|_{F}^2+\operatorname{Tr}\left(\left(\mathbf{A}_0+\mathbf{B}_0\right)\left(\mathbf{A}-\mathbf{A}_0+\mathbf{B}-\mathbf{B}_0\right)\right)-\frac{1}{2}\left\|\mathbf{A}\right\|_{F}^2-\frac{1}{2}\left\|\mathbf{B}\right\|_{F}^2 \buildrel (b) \over\le \operatorname{Tr}\left(\mathbf{A}\mathbf{B}\right) \\
			&\buildrel (c) \over\le \frac{1}{2}\left\|\mathbf{A}+\mathbf{B}\right\|_{F}^2+\frac{1}{2}\left\|\mathbf{A}_0\right\|_{F}^2-\operatorname{Tr}\left(\mathbf{A}_0\mathbf{A}\right)+\frac{1}{2}\left\|\mathbf{B}_0\right\|_{F}^2-\operatorname{Tr}\left(\mathbf{B}_0\mathbf{B}\right)={f}_\text{ub}\left(\mathbf{A},\mathbf{B},\mathbf{A}_0,\mathbf{B}_0\right).
		\end{aligned}
	\end{equation}
\end{figure*}

\subsection{Proposed Solution}
It can be observed that the matrices $\mathbf{V}$ and $\mathbf{W}_\mathrm{B}$, $\mathbf{W}_\mathrm{C}$ are coupled multiplicatively in \eqref{2a}, \eqref{2b}, and \eqref{2c}, and the rank-one constraints \eqref{2f} and \eqref{2i} further increase the difficulty for solving $\mathcal{P}1.1$ directly. In the following, we address these challenges using SCA, penalty method, and SDR technique.

Firstly, we present the following theorem to transform \eqref{2a}, \eqref{2b}, and \eqref{2c} into a more tractable form.

\begin{theorem}\label{Theorem1}  \emph{For any two Hermitian and positive semi-definite matrices} $\mathbf{A}\in \mathbb{H}^M $ \emph{and} $ \mathbf{B}\in \mathbb{H}^M$, \emph{the inequality~\eqref{AB_le} always holds at a given point $\left\{\mathbf{A}_0, \mathbf{B}_0\right\}$, where ${f}_\text{lb}(\cdot)$ and ${f}_\text{ub}(\cdot)$ denote the lower bound concave function and the upper bound convex function of $\operatorname{Tr}(\mathbf{A}\mathbf{B})$, respectively. Both inequalities (b) and (c) are tight when $\mathbf{A}=\mathbf{A}_0$ and~$\mathbf{B}=\mathbf{B}_0$.}
	\begin{proof}
		For matrices $\mathbf{A}$ and $\mathbf{B}$ defined in the theorem, there exists that \begin{equation}\label{AB}
			\operatorname{Tr}\left(\mathbf{A}\mathbf{B}\right)=\frac{1}{2}\left\|\mathbf{A}+\mathbf{B}\right\|_{F}^2-\Big(\frac{1}{2}\left\|\mathbf{A}\right\|_{F}^2+\frac{1}{2}\left\|\mathbf{B}\right\|_{F}^2\Big).\vspace{-0.1cm}
		\end{equation}
		The right-hand side of~\eqref{AB} is in the form of a difference of convex (DC) functions. The global lower bound of the first convex term can be derived by applying the first-order Taylor expansion at the given point $\left\{\mathbf{A}_0, \mathbf{B}_0\right\}$, leading to
		\begin{equation}\label{AB_1}
			\begin{aligned}
				\frac{1}{2}\left\|\mathbf{A}+\mathbf{B}\right\|_{F}^2 &\ge \frac{1}{2}\left\|\mathbf{A}_0+\mathbf{B}_0\right\|_{F}^2\\&+\operatorname{Tr}\left(\left(\mathbf{A}_0+\mathbf{B}_0\right)\left(\mathbf{A}-\mathbf{A}_0+\mathbf{B}-\mathbf{B}_0\right)\right).
			\end{aligned}
		\end{equation} 
		Similarly, the lower bound of the second convex term can be obtained based on
		\begin{equation}\label{AB_2}
			\begin{aligned}
					\frac{1}{2}\left\|\mathbf{X}\right\|_{F}^2 &\ge 	\frac{1}{2}\left\|\mathbf{X}_0\right\|_{F}^2+\operatorname{Tr}\left(\mathbf{X}_0\left(\mathbf{X}-\mathbf{X}_0\right)\right)\\
				&=\operatorname{Tr}\left(\mathbf{X}_0\mathbf{X}\right)-\frac{1}{2}\left\|\mathbf{X}_0\right\|_{F}^2,
			\end{aligned} 
		\end{equation} 
		where $\left\{\mathbf{X}, \mathbf{X}_0\right\}$ could be $\left\{\mathbf{A}, \mathbf{A}_0\right\}$ or $\left\{\mathbf{B}, \mathbf{B}_0\right\}$. Then, substituting~\eqref{AB_1} into~\eqref{AB} leads to the inequality (b) in~\eqref{AB_le}, while substituting~\eqref{AB_2} yields the inequality (c).
	\end{proof}
\end{theorem}

We treat $\mathbf{V}$ and $\mathbf{G}_\mathbf{B}\mathbf{W}_\mathbf{B}\mathbf{G}_\mathbf{B}^\mathrm{H}$ as $\mathbf{A}$ and $\mathbf{B}$, respectively, and apply the SCA method to solve the proposed problem. Then, given $\left\{\mathbf{V}^{(n)}, \mathbf{W}_\mathrm{B}^{(n)}, \mathbf{W}_\mathrm{C}^{(n)}\right\}$ in the $n$-th iteration, the lower bound of the objective in \eqref{2a} can be derived based on \textbf{Theorem \ref{Theorem1}} as
\begin{small} 
\begin{equation}
\overline{\operatorname{Tr}(\mathbf{V}\mathbf{G}_\mathrm{B}\mathbf{W}_\mathrm{B}\mathbf{G}_\mathrm{B}^\mathrm{H})} = f_\text{lb}\left(\mathbf{V},\mathbf{G}_\mathrm{B}\mathbf{W}_\mathrm{B}\mathbf{G}_\mathrm{B}^\mathrm{H},\mathbf{V}^{(n)},\mathbf{G}_\mathrm{B}\mathbf{W}_\mathrm{B}^{(n)}\mathbf{G}_\mathrm{B}^\mathrm{H}\right).	
\end{equation} 
\end{small}Furthermore, we apply inequalities (b) and (c) in \eqref{AB_le} to the right-hand side and left-hand side of the constraint \eqref{2b}, respectively, and it can be reformulated as
\begin{equation}\label{c1}
	\begin{aligned}
		&\rho {f}_\text{ub}\left(\mathbf{V},\mathbf{G}_\mathrm{C}\mathbf{W}_\mathrm{B}\mathbf{G}_\mathrm{C}^\mathrm{H},\mathbf{V}^{(n)},\mathbf{G}_\mathrm{C}\mathbf{W}_\mathrm{B}^{(n)}\mathbf{G}_\mathrm{C}^\mathrm{H}\right) + \rho\sigma_\mathrm{C}^2 \\
		&\leq
		{f}_\text{lb}\left(\mathbf{V},\mathbf{G}_\mathrm{C}\mathbf{W}_\mathrm{C}\mathbf{G}_\mathrm{C}^\mathrm{H},\mathbf{V}^{(n)},\mathbf{G}_\mathrm{C}\mathbf{W}_\mathrm{C}^{(n)}\mathbf{G}_\mathrm{C}^\mathrm{H}\right).
	\end{aligned}
\end{equation}
Similarly, the constraint \eqref{2c} can be rewritten as
\begin{equation}\label{c2}
	\begin{aligned}
		&\rho {f}_\text{ub}\left(\mathbf{V},\mathbf{G}_\mathrm{B}\mathbf{W}_\mathrm{B}\mathbf{G}_\mathrm{B}^\mathrm{H},\mathbf{V}^{(n)},\mathbf{G}_\mathrm{B}\mathbf{W}_\mathrm{B}^{(n)}\mathbf{G}_\mathrm{B}^\mathrm{H}\right) + \rho\sigma_\mathrm{B}^2 \\
		&\leq
		{f}_\text{lb}\left(\mathbf{V},\mathbf{G}_\mathrm{B}\mathbf{W}_\mathrm{C}\mathbf{G}_\mathrm{B}^\mathrm{H},\mathbf{V}^{(n)},\mathbf{G}_\mathrm{B}\mathbf{W}_\mathrm{C}^{(n)}\mathbf{G}_\mathrm{B}^\mathrm{H}\right).
	\end{aligned}
\end{equation}
After the above transformation, the objective function \eqref{2a} is converted into a concave form, and \eqref{2b} and \eqref{2c} are approximated as convex constraints.

Next, we utilize the penalty method to tackle the rank-one constraint \eqref{2f}. It is worth noting that for any $\mathbf{V} \in \mathbb{H}^{M}$ with $\mathbf{V} \succeq 0$, the inequality $\left|\mathbf{V}\right|{*} - \left|\mathbf{V}\right|_{2} \geq 0$ always holds, with equality if and only if $\mathbf{V}$ is a rank-one matrix. Therefore, the non-convex constraint \eqref{2f} is equivalent to the following equality constraint:
\begin{equation}\label{rank_one}
	\left\|\mathbf{V}\right\|_{*}-\left\|\mathbf{V}\right\|_{2}=0,
\end{equation}
where $\left\|\mathbf{V}\right\|_{*}$ and $\left\|\mathbf{V}\right\|_{2}$ denote the nuclear norm and the spectral norm, respectively.
To address this nonlinear equality constraint, a penalty-based approach is adopted. Specifically, a penalty factor $\eta>0$ is introduced, and the constraint  \eqref{rank_one} is relaxed by incorporating it into the objective function as a penalty term, resulting in the following reformulated optimization problem:
\begin{subequations}\label{P1.3}
	\begin{align}
    (\mathcal{P}1.2)~~&\mathop {{\rm{Max}}}\limits_{\mathbf{W}_\mathrm{B},\mathbf{W}_\mathrm{C},\mathbf{V}} 
	\overline{\operatorname{Tr}(\mathbf{V}\mathbf{G}_\mathrm{B}\mathbf{W}_\mathrm{B}\mathbf{G}_\mathrm{B}^\mathrm{H})} - \eta \big(\big\|\mathbf{V}\big\|_* -\left\|\mathbf{V}\right\|_{2} \big),\\
    ~~\text{s.t.} ~~~~
	& \eqref{c1}, \eqref{c2}, \eqref{2d},\eqref{2e},\eqref{2g} \sim \eqref{2j}.
    \end{align}
\end{subequations}Since the introduced penalty term is non-convex with a DC structure, we linearize the spectral norm $\left\|\mathbf{V}\right\|_2$ via its first-order Taylor approximation at the given point $\mathbf{V}^{(n)}$. Consequently, a convex upper bound can be obtained as
\begin{equation}
		\big\|\mathbf{V}\big\|_* - \big\|\mathbf{V}\big\|_2 
		 \leq \big\|\mathbf{V}\big\|_* - \overline{\mathbf{V}}^{(n)},	
\end{equation}
where
{\small \begin{equation*}
	\overline{\mathbf{V}}^{(n)} \triangleq \left\|\mathbf{V}^{(n)}\right\|_2 + \mathrm{Tr}\left[\overline{\mathbf{u}}\left(\mathbf{V}^{(n)}\right)\left(\overline{\mathbf{u}}\left(\mathbf{V}^{(n)}\right)\right)^{\mathrm{H}} \left(\mathbf{V} - \mathbf{V}^{(n)}\right)\right],
\end{equation*}}and $\overline{\mathbf{u}}\left(\mathbf{V}^{(n)}\right)$ represents the eigenvector associated with the largest eigenvalue of $\mathbf{V}^{(n)}$.

Now, the main challenge in solving \eqref{P1.3} stems from the non-convex rank-one constraints on $\mathbf{W}_\mathrm{B}$ and $\mathbf{W}_\mathrm{C}$. To address this issue, we adopt the SDR method, i.e., solving the formulated problem with ignoring the rank-one constraint \eqref{2i}. The feasibility and tightness of this relaxation are guaranteed by the following theorem.

\begin{theorem}\label{Theorem2}  \emph{Without loss of optimality, the solutions $\{\mathbf{W}_\mathrm{B}, \mathbf{W}_\mathrm{C}\}$ obtained for the problem $\mathcal{P}1.2$ without the constraint \eqref{2i}, always satisfy $\operatorname{Rank}(\mathbf{W}_\mathrm{B}) =1$ and $\operatorname{Rank}(\mathbf{W}_\mathrm{C}) =1$
This conclusion can be rigorously proved based on Slater’s constraint qualification and the Karush-Kuhn-Tucker (KKT) conditions. Detailed derivations can be found in \cite{9570143,9183907}.}
\end{theorem}
Based on the aforementioned SCA, penalty method and SDR, the original problem $\mathcal{P}1.1$ is finally transformed into

\begin{subequations}\label{P1.4}
	\begin{align}
		(\mathcal{P}1.3)	~~&\mathop {{\rm{Max}}}\limits_{\mathbf{W}_\mathrm{B},\mathbf{W}_\mathrm{C},\mathbf{V}} 
		\overline{\operatorname{Tr}(\mathbf{V}\mathbf{G}_\mathrm{B}\mathbf{W}_\mathrm{B}\mathbf{G}_\mathrm{B}^\mathrm{H})} - \eta \big(\big\|\mathbf{V}\big\|_* - \overline{\mathbf{V}}^{(n)} \big),\label{1.3a} \\
		~~\text{s.t.} ~~~~
		& \eqref{c1}, \eqref{c2},\eqref{2d},\eqref{2e},\eqref{2g},\eqref{2h},\eqref{2j},
	\end{align}
\end{subequations}
which is a standard concave semidefinite program (SDP) and can be efficiently solved using existing solvers such as CVX.
It is verified that when $\eta \to \infty$, the solution of problem \eqref{P1.4} always satisfies the rank-one constraint \eqref{2f}~\cite{ben1997penalty}. However, an excessively large value of $\eta$ may cause the penalty term to become dominant in the objective function, which adversely affects the optimization performance. To avoid this, $\eta$ is initialized to a small value, and a two-loop iterative algorithm is employed. In the inner loop, problem \eqref{P1.4} is solved with a given $\eta$. In the outer loop, $\eta$ is updated according to $\eta^{(i+1)} = \omega \eta^{(i)}$, where $\omega > 1$ represents the scaling factor, and $i$ is the iteration index of the outer loop. The termination condition of the algorithm is that the value of the penalty term falls below a predefined threshold, denoted by $\mu$.

\begin{algorithm}[tbp]
	\caption{Joint optimization algorithm for $\mathcal{P}1$}
	\label{algorithm1}
	\begin{algorithmic}[1] 
		\STATE Initialize feasible points $\big\{\mathbf{V}^{(0)}, \mathbf{W}_\mathrm{B}^{(0)}, \mathbf{W}_\mathrm{C}^{(0)}\big\}$, and the penalty factor $\eta^{(0)}$.
		\STATE Set the iteration index $i=0$ for outer loop.
		\STATE \textbf{Repeat}
		\STATE \hspace{8pt} Set the iteration index $n=0$ for inner loop.
		\STATE \hspace{8pt} \textbf{Repeat}
		\STATE \hspace{20pt} Solve problem~\eqref{P1.4} with $\big\{\mathbf{V}^{(n)}, \mathbf{W}_\mathrm{B}^{(n)}, \mathbf{W}_\mathrm{C}^{(n)}\big\}$, $\eta^{(i)}$.
		\STATE \hspace{20pt} Update $\big\{\mathbf{V}^{(n+1)}, \mathbf{W}_\mathrm{B}^{(n+1)}, \mathbf{W}_\mathrm{C}^{(n+1)}\big\}$, and $n=n+1$.
		\STATE \hspace{7pt} \textbf{Until} the predefined accuracy or maximum iteration number is reached.
		\STATE \hspace{7pt} Update  $\big\{\mathbf{V}^{(0)}, \mathbf{W}_\mathrm{B}^{(0)}, \mathbf{W}_\mathrm{C}^{(0)}\big\}$ =  $\big\{\mathbf{V}^{(n)}, \mathbf{W}_\mathrm{B}^{(n)}, \mathbf{W}_\mathrm{C}^{(n)}\big\}$.
		\STATE \hspace{7pt} Update $\eta^{(i+1)} = \omega \eta^{(i)}$, and $i=i+1$.
		\STATE \textbf{Until} $ \big\|\mathbf{V}\big\|_*^{(i)} - \big\|\mathbf{V}\big\|_2^{(i)} \leq \mu$.
	\end{algorithmic}
\end{algorithm}

The details of the proposed algorithm are summarized in \textbf{Algorithm~\ref{algorithm1}}. As the objective function~\eqref{1.3a} monotonically increases during the inner iterations and is upper-bounded, the algorithm is guaranteed to converge to a stationary point of the original problem. The computational complexity of the proposed algorithm is mainly attributed to solving the SDP problem \eqref{P1.4}, whose complexity is $\mathcal{O}(N_t^{3.5} + M^{3.5})$ using the interior-point method~\cite{5447068}. Therefore, the overall complexity of \textbf{Algorithm~\ref{algorithm1}} is $\mathcal{O}(I_o I_i (N_t^{3.5} + M^{3.5}))$, with $I_o$ and $I_i$ denoting the numbers of outer and inner iterations required for convergence, respectively. 

\begin{figure*}[b] 
	\rule{\textwidth}{0.4pt} 
	\setcounter{equation}{34}
	\begin{equation}
		\label{CRB}
		\mathrm{CRB}(\hat{\theta}_\mathrm{W}) = \frac{\sigma_A^2\operatorname{Tr}(\mathbf{A}(\hat{\theta}_\mathrm{W})^\mathrm{H}\mathbf{A}(\hat{\theta}_\mathrm{W})\mathbf{W}) }{2|\alpha|^2L\left( \operatorname{Tr}(\dot{\mathbf{A}}(\hat{\theta}_\mathrm{W})^\mathrm{H}\dot{\mathbf{A}}(\hat{\theta}_\mathrm{W})\mathbf{W} )\operatorname{Tr}(\mathbf{A}(\hat{\theta}_\mathrm{W})^\mathrm{H}\mathbf{A}(\hat{\theta}_\mathrm{W})\mathbf{W}) - |\operatorname{Tr}(\dot{\mathbf{A}}(\hat{\theta}_\mathrm{W})^\mathrm{H}\mathbf{A}(\hat{\theta}_\mathrm{W})\mathbf{W})|^2 \right)}
	\end{equation}
	\setcounter{equation}{37}
	\begin{equation}
		\label{delta_a}
		\| \Delta\mathbf{a}_\mathrm{A}(\hat{\theta}_\mathrm{W},\Delta\theta_\mathrm{W})\|^2 = \sum_{n=1}^{N_t} 2- 2\cos \Big(\frac{2\pi(n-1)d}{\lambda}\big(\sin\big(\hat{\theta}_\mathrm{W}+\Delta\theta_\mathrm{W}\big)-\sin\hat{\theta}_\mathrm{W}\big)\Big)
	\end{equation}
\end{figure*}

\section{Covert Communication Under Unknown Warden Location}\label{imperfect}
In this section, we further investigate a more practical and challenging scenario where Alice has no prior knowledge of Willie’s location, and thus the perfect CSI in ~\eqref{h_AW} is unavailable. In such cases, we propose to exploit Alice’s sensing capability to estimate Willie’s location, based on which an imperfect CSI can be obtained, with the CSI uncertainty characterized by the sensing accuracy. Accordingly, Alice is supposed to fulfill three key functions: localizing the warden through sensing, achieving covert transmission to Bob, and supporting reliable communication services for the public user, Carol.

In the following, we first derive the expression of the sensing-based imperfect channel vector, denoted by $\mathbf{\overline{h}}_\mathrm{AW}$, based on which the MDEP at Willie is analyzed. Subsequently, a robust covert rate maximization problem is formulated, subject to the QoS constraint for Carol, the CRB-based sensing accuracy constraint, and the covertness constraint for Bob. Finally, the proposed problem is solved with the aid of the Schur complement and the S-procedure.

\subsection{Sensing-based Imperfect CSI of Willie}

\setcounter{equation}{33}
Recall that the channel from Alice to Willie is modeled as
$\mathbf{h}_{\mathrm{AW}} = \sqrt{\frac{\beta_\mathrm{W}}{d_{\mathrm{AW}}^{\alpha_\mathrm{W}}}}\mathbf{a}_\mathrm{A}(\theta_\mathrm{W})$, where the large-scale fading coefficient $\beta_\mathrm{W}$ and the path-loss exponent $\alpha_\mathrm{W}$ can be obtained in advance through field measurements or radio maps, as suggested in~\cite{10592837,9933849}. 
Moreover, the distance $d_{\mathrm{AW}}$ has a negligible effect on the covertness metric due to Willie’s energy-ratio-based detection rule, which will be elaborated in the next subsection.
Consequently, the channel vector can be treated as a function of the angle of Willie, and the problem of acquiring the CSI can thus be transformed into an angle estimation problem for Willie.

We proposes leveraging Alice's radar sensing capability to obtain an estimate of Willie's angle, denoted by $\hat{\theta}_\mathrm{W}$. Specifically, prior to communication, Alice transmits an omnidirectional waveform and receives the echoes reflected by both the RIS and Willie. As the RIS is a cooperative node with a known location, an estimate of Willie's angle can be obtained by eliminating the echo signal from the RIS and applying the combined Capon and approximate maximum likelihood (CAML) technique~\cite{10227884,4655353}. After acquiring $\hat{\theta}_\mathrm{W}$, the system performs sensing and communication simultaneously, i.e., the signal defined in~\eqref{signal} is used both for communication with Bob and Carol and for sensing Willie.
Without affecting the analysis or optimization, we assume the radar cross section (RCS) to be~1 for simplicity~\cite{10364735}. Then, the echo signal reflected by Willie and received at Alice in the $l$-th time-domain snapshot is given as
\begin{equation}\label{echo}
	\mathbf{y}_\mathrm{A}[l] = {\frac{\beta_\mathrm{B(C)}}{d_{\mathrm{AW}}^{\alpha_\mathrm{W}}}} \mathbf{b}_\mathrm{A}(\hat{\theta}_\mathrm{W}) \mathbf{a}_\mathrm{A}(\hat{\theta}_\mathrm{W})^\mathrm{H}\mathbf{x}[l]  + \mathbf{n}_\mathrm{A}[l] ,
\end{equation}  
where $\mathbf{b}_\mathrm{A}(\hat{\theta}_\mathrm{W}) \in \mathbb{C}^{N_r \times 1}$ denotes the receiving steering vector at Alice, and $\mathbf{n}_\mathrm{A}[l] \sim \mathcal{CN}(0,\sigma_\mathrm{A}^2)$ is the noise at Alice. 
For convenience, we define $\alpha = {\frac{\beta_\mathrm{B(C)}}{d_{\mathrm{AW}}^{\alpha_\mathrm{W}}}}$
and $\mathbf{A}(\hat{\theta}_\mathrm{W})= \mathbf{b}_\mathrm{A}(\hat{\theta}_\mathrm{W}) \mathbf{a}_\mathrm{A}(\hat{\theta}_\mathrm{W})^\mathrm{H}$, and then \eqref{echo} can be expressed as $	\mathbf{y}_\mathrm{A}[l]= \alpha\mathbf{A}(\hat{\theta}_\mathrm{W})\mathbf{x}[l]+ \mathbf{n}_\mathrm{A}[l]$.

The accuracy of the estimated angle is quantified in terms of the CRB~\cite{9652071}, which is given in \eqref{CRB}, with $\dot{\mathbf{A}}(\hat{\theta}_\mathrm{W}) = \frac{\partial{\mathbf{A}}(\hat{\theta}_\mathrm{W})}{\partial\hat{\theta}_\mathrm{W}}$ and $\mathbf{W} = \mathbf{W}_\mathrm{B} + \mathbf{W}_\mathrm{C}$ denoting the covariance matrix of the transmitted signal.
Assume that the estimation error, denoted by $\Delta\theta_\mathrm{W}$, follows a Gaussian distribution with zero mean and a variance of $\mathrm{CRB}(\hat{\theta}_\mathrm{W})$, i.e., $\Delta\theta_\mathrm{W} \sim \mathcal{CN}(0, \mathrm{CRB}(\hat{\theta}_\mathrm{W}))$. 
Therefore, the real value of Willie's angle lies within the interval  
$ \Theta = \left[ \hat{\theta}_\mathrm{W} - 3\sqrt{\mathrm{CRB}(\hat{\theta}_\mathrm{W})},\ \hat{\theta}_\mathrm{W} + 3\sqrt{\mathrm{CRB}(\hat{\theta}_\mathrm{W})} \right]
$ with a probability of approximately 0.9973.

\setcounter{equation}{35}
So far, an imperfect channel of Willie can be obtained based on the estimated angle $\hat{\theta}_\mathrm{W}$, which is expressed as 
\begin{equation}\label{h_im}
\overline{\mathbf{h}}_{\mathrm{AW}}({\theta}_\mathrm{W} ) =\sqrt{\frac{\beta_\mathrm{W}}{d_\mathrm{AW}^{\alpha_\mathrm{W}}}} \mathbf{a}_A({\theta}_\mathrm{W} ), 	\theta_\mathrm{W} \in \Theta.
\end{equation}
To specifically characterize the uncertainty of CSI caused by the estimation error, we express the steering vector $\mathbf{a}_\mathrm{A}$ as
\begin{equation}\label{a_im}
	\mathbf{a}_\mathrm{A}({\theta}_\mathrm{W}) = \mathbf{a}_\mathrm{A}(\hat{\theta}_\mathrm{W} + \Delta\theta_\mathrm{W}) = \mathbf{a}_\mathrm{A}(\hat{\theta}_\mathrm{W}) +  \Delta\mathbf{a}_\mathrm{A}(\hat{\theta}_\mathrm{W}, \Delta\theta_\mathrm{W}),	
\end{equation}
where $\Delta\mathbf{a}_\mathrm{A}(\hat{\theta}_\mathrm{W},\Delta\theta_\mathrm{W})$ is exactly the uncertainty term, and its squared Euclidean norm can be expressed as \eqref{delta_a}, with the aid of Euler's theorem and trigonometric transformations.
To make equation \eqref{delta_a} easier to handle, we take the following approximation assuming a sufficiently small $\Delta\theta_\mathrm{W}$, that is,   
\setcounter{equation}{38}
\begin{equation}
	\label{sin}
	\sin\big(\hat{\theta}_\mathrm{W}+\Delta\theta_W\big) \approx \sin(\hat{\theta}_\mathrm{W})+\cos(\hat{\theta}_\mathrm{W})\Delta\theta_\mathrm{W}.
\end{equation}
Substituting \eqref{sin} into \eqref{delta_a} and using $\cos(x) \approx 1- \frac{x^2}{2}$, it can be derived that
{\small \begin{equation}\label{appro}
	\| \Delta \mathbf{a}_\mathrm{A}(\hat{\theta}_\mathrm{W},\Delta\theta_\mathrm{W})\|^2 \approx 
	\sum_{n=1}^{N_t} \Big(\frac{2\pi(n-1)d}{\lambda}\cos(\hat{\theta}_\mathrm{W})\Delta\theta_\mathrm{W})\Big)^2.
\end{equation}}
Since $|\Delta\theta_\mathrm{W}| \leq 3\sqrt{\mathrm{CRB}(\hat{\theta}_\mathrm{W})}$, we have
\begin{footnotesize}
\begin{equation}\label{bound}
	\begin{aligned}
		\| \Delta \mathbf{a}_\mathrm{A}(\hat{\theta}_\mathrm{W},\Delta\theta_\mathrm{W})\|^2  
		&\leq \sum_{n=1}^{N_t} \Big(\frac{2\pi(n-1)d}{\lambda}\cos(\hat{\theta}_\mathrm{W})3\sqrt{\mathrm{CRB}(\hat{\theta}_\mathrm{W})}\Big)^2 \\ &\triangleq K \cdot \mathrm{CRB}(\hat{\theta}_\mathrm{W}),	
	\end{aligned}
\end{equation}
\end{footnotesize}where $ K \triangleq  \left( \frac{6\pi d}{\lambda} \cos(\hat{\theta}_\mathrm{W}) \right)^2 \frac{(N_t - 1) N_t (2N_t - 1)}{6} $.
The sensing based imperfect channel of Willie can be finally given as
\begin{equation}\label{imperfect_h}
	\overline{\mathbf{h}}_{\mathrm{AW}}({\theta}_\mathrm{W}) =\sqrt{\frac{\beta_\mathrm{W}}{d_{\mathrm{AW}}^{\alpha_\mathrm{W}}}} \big( \mathbf{a}_A(\hat{\theta}_\mathrm{W} ) +  \Delta\mathbf{a}_A(\hat{\theta}_\mathrm{W},\Delta\theta_\mathrm{W})\big), 
\end{equation}
with the uncertainty term bounded by \eqref{bound}.

\subsection{Covertness Analysis and Problem Formulation}
Based on the MDEP detection and the covert communication requirement given in Sections~\ref{C} and~\ref{D}, the covertness constraint under the proposed sensing-based imperfect CSI is expressed as

\begin{equation}
	\label{covert}
	\frac{\overline{\lambda}_1(\theta_\mathrm{W})}{\overline{\lambda}_0(\theta_\mathrm{W})} \le x_2(\epsilon), \forall \theta_\mathrm{W} \in \Theta,
\end{equation}
where
\begin{equation}\label{lamba}
	\begin{cases}
	\overline{\lambda}_0(\theta_\mathrm{W})= \alpha \operatorname{Tr}(\mathbf{a}_\mathrm{A}(\theta_\mathrm{W})\mathbf{a}_\mathrm{A}(\theta_\mathrm{W})^\mathrm{H}\mathbf{W}_\mathrm{C})+\sigma_\mathrm{W}^2, \\
	\overline{\lambda}_1(\theta_\mathrm{W})= \alpha \operatorname{Tr}(\mathbf{a}_\mathrm{A}(\theta_\mathrm{W})\mathbf{a}_\mathrm{A}(\theta_\mathrm{W})^\mathrm{H}\mathbf{W})+\sigma_\mathrm{W}^2.
	\end{cases}
\end{equation}
It is observed that the large-scale fading term $\alpha={\frac{\beta_\mathrm{W}}{d_{\mathrm{AW}}^{\alpha_\mathrm{W}}}}$ appears in both the numerator and the denominator of \eqref{covert}, and can thus be neglected when the noise power is sufficiently low. This explains why the estimation for Willie’s distance, $d_\mathrm{AW}$, is not considered in this paper. 

Substituting \eqref{a_im} into \eqref{lamba}, we further have
\begin{equation}\label{lambda1}
	\begin{cases}
		\overline{\lambda}_0(\theta_\mathrm{W}) = 
		\alpha 
		\mathrm{Q}(\hat{\theta}_\mathrm{W}, \mathbf{W}_\mathrm{C}) +\sigma_\mathrm{W}^2,\\
		\overline{\lambda}_1(\theta_\mathrm{W}) =		
		\alpha
		\mathrm{Q}(\hat{\theta}_\mathrm{W}, \mathbf{W})
		+\sigma_\mathrm{W}^2,
	\end{cases}	
\end{equation}
where $\mathrm{Q}(\hat{\theta}_\mathrm{W},\mathbf{X})$ is given as \eqref{Q} with $\mathbf{X} \in \mathbb{C}^{N_t}$ being $\mathbf{W}$ or $\mathbf{W}_\mathrm{C}$.  
\setcounter{equation}{46}
As a result, the covertness constraint \eqref{covert} can be rewritten as 
{\footnotesize \begin{equation}\label{c_im}
	\frac{\alpha\mathrm{Q}(\hat{\theta}_\mathrm{W},\mathbf{W})+\sigma_\mathrm{W}^2}{\alpha\mathrm{Q}(\hat{\theta}_\mathrm{W},\mathbf{W}_\mathrm{C})+\sigma_\mathrm{C}^2}  \leq  x_2(\epsilon), \|\
    \Delta\mathbf{a}_\mathrm{A}(\hat{\theta}_\mathrm{W}, \Delta\theta_\mathrm{W})\|^2 \leq  K \cdot \mathrm{CRB}(\hat{\theta}_\mathrm{W}).	
\end{equation}}

\begin{figure*}[b] 
	\rule{\textwidth}{0.4pt} 
	\setcounter{equation}{45}
	\begin{equation}\label{Q}
		\mathrm{Q}(\hat{\theta}_\mathrm{W},\mathbf{X})  =
		\Delta{\mathbf{a}}_\mathrm{A}(\hat{\theta}_\mathrm{W}, \Delta\theta_\mathrm{W})^\mathrm{H}\mathbf{X}\Delta{\mathbf{a}}_\mathrm{A}(\hat{\theta}_\mathrm{W}, \Delta\theta_\mathrm{W}) + 	2\operatorname{Re}({\mathbf{a}}_\mathrm{A}(\hat{\theta}_\mathrm{W})^\mathrm{H}\mathbf{X}\Delta{\mathbf{a}}_\mathrm{A}(\hat{\theta}_\mathrm{W}, \Delta\theta_\mathrm{W})) +{{\mathbf{a}}_\mathrm{A}(\hat{\theta}_\mathrm{W})^\mathrm{H}\mathbf{X}{\mathbf{a}}_\mathrm{A}(\hat{\theta}_\mathrm{W})}
	\end{equation}
\end{figure*}

\setcounter{equation}{47}
Our objective is to maximize the communication rate for Bob, subject to a robust covertness constraint formulated under the proposed  sensing-based imperfect CSI model. Meanwhile, the QoS requirement for Carol is also considered. This  optimization problem serves as an extension of the known-location case discussed in Section~\ref{perfect}, and, based on problem $\mathcal{P}1.3$, it can be formulated as
\begin{subequations}\label{P2}
	\begin{align}
		(\mathcal{P}2)	~~&\mathop {{\rm{Max}}}\limits_{\mathbf{W}_\mathrm{B},\mathbf{W}_\mathrm{C},\mathbf{V}} 
		\overline{\operatorname{Tr}(\mathbf{V}\mathbf{G}_\mathrm{B}\mathbf{W}_\mathrm{B}\mathbf{G}_\mathrm{B}^\mathrm{H})} - \eta \big(\big\|\mathbf{V}\big\|_* - \overline{\mathbf{V}}^{(n)} \big),\label{3a} \\
		~~\text{s.t.} ~~~~
		& \mathrm{CRB}(\hat{\theta}_\mathrm{W}) \leq \Gamma_\mathrm{CRB}, \label{3b}\\
		& \eqref{c1}, \eqref{c2}, \eqref{2d},\eqref{2e},\eqref{2g},\eqref{2h}, \eqref{c_im}
	\end{align}
\end{subequations}
where \eqref{3b} imposes a constraint on Alice's sensing accuracy, with $\Gamma_\mathrm{CRB}$  representing the maximum tolerable CRB.
\begin{remark}\label{remark1} 
\emph{
The proposed sensing-based imperfect CSI model, shown as \eqref{imperfect_h}, shares the same mathematical form with conventional norm-bounded imperfect CSI models~\cite{10546987,10663964,10497104}, i.e., $\mathbf{h} = \hat{\mathbf{h}} + \Delta\mathbf{h}$, where~$|\Delta \mathbf{h}|^2 \leq \delta$, and~$\delta$ represents the CSI error bound.
The key difference lies in how~$\delta$ is determined. Unlike being a given constant in traditional norm-bounded imperfect CSI models,  $\delta$ is formulated as a function of the sensing performance, specifically characterized by CRB in our work.
The proposed CRB-bounded imperfect CSI model of Willie introduces a subtle connection between sensing performance and communication covertness, as given in \eqref{c_im}.}
\emph{
This feature makes a dual yet contradictory effect of sensing on communication.
Specifically, from a robustness perspective, improving the sensing accuracy (i.e., reducing the CRB) leads to reduced channel uncertainty, which in turn relaxes the covertness constraint \eqref{c_im} and facilitates improved covert communication performance.
However, it should be noted that there is an inherent competition for the beamforming resource at Alice between sensing and communication. Excessively pursuing a lower CRB would come at the expense of the beam gain directed toward the RIS, thereby reducing the achievable covert rate.
Therefore, the relationship between covert communication and sensing is not simply monotonic. An optimal beamforming design is expected to achieve a reasonable CRB to balance such dual and conflicting effects. }
\end{remark}

\subsection{Proposed Solutions}
The difficulty in solving problem $\mathcal{P}2$ mainly lies in the non-convex constraints \eqref{3b} and \eqref{c_im}, which will be addressed  using the Schur complement and the S-procedure, respectively.

Firstly, to tackle the constraint~\eqref{3b}, the following theorem is established according to the Schur complement condition~\cite{Brezinski2005}.
\begin{theorem} \label{theorem3} 
	\emph{Consider a Hermitian matrix}
	\[
	\mathbf{M} = 
	\begin{bmatrix}
		\mathbf{A} & \mathbf{B} \\
		\mathbf{B}^\mathrm{H} & \mathbf{C}
	\end{bmatrix},
	\]
\emph{where $\mathbf{A}$, $\mathbf{B}$, and $\mathbf{C}$ are matrices of appropriate dimensions. If $\mathbf{C}$ is invertible, the following two propositions are equivalent:}
	\begin{enumerate}
		\item $\mathbf{M} \succeq 0$;
		\item $\mathbf{C} \succeq 0$ \emph{and} $\mathbf{A}-\mathbf{B} \mathbf{C}^{-1} \mathbf{B}^\mathrm{H}  \succeq  0$.
	\end{enumerate}
\end{theorem}
\setcounter{equation}{48}

Based on \textbf{Theorem~3}, since $\operatorname{Tr}({\mathbf{A}}^\mathrm{H}(\hat{\theta}_\mathrm{W}) {\mathbf{A}}(\hat{\theta}_\mathrm{W})\mathbf{W})$ is guaranteed to be non-negative and invertible,  the constraint~ \eqref{3b} can be equivalently transformed into
\begin{equation} \label{sensing}
	\mathbf{U} - \begin{bmatrix}
		t_0 & 0 \\
		0 & 0
	\end{bmatrix} \succeq 0, 	
\end{equation}
where
{\small \[ \mathbf{U} = \begin{bmatrix}
	\operatorname{Tr}(\dot{\mathbf{A}}^\mathrm{H}(\hat{\theta_\mathrm{W}}) \dot{\mathbf{A}}(\hat{\theta_\mathrm{W}})\mathbf{W}) 
	& 
	\operatorname{Tr}(\dot{\mathbf{A}}^\mathrm{H}(\hat{\theta_\mathrm{W}}){\mathbf{A}}(\hat{\theta_\mathrm{W}})\mathbf{W}) \\
	
	\operatorname{Tr}({\mathbf{A}}^\mathrm{H}(\hat{\theta_\mathrm{W}}) \dot{\mathbf{A}}(\hat{\theta_\mathrm{W}})\mathbf{W}) & 
	\operatorname{Tr}({\mathbf{A}}^\mathrm{H}(\hat{\theta_\mathrm{W}}) {\mathbf{A}}(\hat{\theta_\mathrm{W}})\mathbf{W})
\end{bmatrix},
\] }
and $t_0 = \frac{\sigma_\mathrm{A}^2}{2|\alpha|^2L\Gamma_\mathrm{CRB}}$. 
It is seen that following the above transformation, the original non-convex fractional constraint~\eqref{3b} is converted into the more tractable convex form~\eqref{sensing}.

Next, we present the following theorem to handle the uncertainty sets in \eqref{c_im}.
\begin{theorem}\label{Theorem4}  \emph{Consider two quadratic functions with respect to $\mathbf{x} \in \mathbb{C}^{N\times 1}$:
\begin{equation*}
	f_1(\mathbf{x}) = \mathbf{x}^\mathrm{H}\mathbf{A}\mathbf{x} + 2 \operatorname{Re}(\mathbf{b}^\mathrm{H}\mathbf{x}) + c, 
\end{equation*}
\begin{equation*}	
	f_2(\mathbf{x}) =\mathbf{x}^\mathrm{H}\mathbf{Q}\mathbf{x},
\end{equation*}
where $\mathbf{A},\mathbf{Q} \in \mathbb{C}^{N\times N}$ are Hermitian matrices,  $\mathbf{b} \in \mathbb{C}^{N\times 1}$, and $c$ is a scalar.
The condition 
\[	f_2(\mathbf{x}) \leq \delta \implies 	f_1(\mathbf{x}) \geq 0\]
holds for all $\mathbf{x}$, if there exists a scalar \(\zeta \geq 0\) such that the following linear matrix inequality (LMI) is satisfied:
\begin{equation}
	\begin{bmatrix}
		\mathbf{A}+\zeta\mathbf{Q} & \mathbf{b}\\
		\mathbf{b}^\mathrm{H}  & c - \zeta \delta
	\end{bmatrix} \succeq 0.	
\end{equation}
\begin{proof}
The proof is a straightforward application of the S-Procedure.
\end{proof}}
\end{theorem}

\begin{figure*}[b] 
	\rule{\textwidth}{0.4pt} 
	\setcounter{equation}{50}
	\begin{equation}\label{cc}
		\small
		\Delta{\mathbf{a}}_\mathrm{A}(\hat{\theta}_\mathrm{W}, \Delta\theta_\mathrm{W})^\mathrm{H}\mathbf{M}\Delta{\mathbf{a}}_\mathrm{A}(\hat{\theta}_\mathrm{W}, \Delta\theta_\mathrm{W}) +  2\operatorname{Re}\big({\mathbf{a}}_\mathrm{A}(\hat{\theta}_\mathrm{W})^\mathrm{H}\mathbf{M}\Delta{\mathbf{a}}_\mathrm{A}\Delta{\mathbf{a}}_\mathrm{A}(\hat{\theta}_\mathrm{W}, \Delta\theta_\mathrm{W})\big)  +{\mathbf{a}}_\mathrm{A}(\hat{\theta}_\mathrm{W})^\mathrm{H}\mathbf{M}{\mathbf{a}}_\mathrm{A}(\hat{\theta}_\mathrm{W}) + \frac{d_{\mathrm{AW}}^{\alpha_\mathrm{W}}\sigma_\mathrm{W}^2(x_2(\epsilon)-1)}{\beta_\mathrm{W}} \geq 0	
	\end{equation}
\end{figure*}

Given \eqref{Q}, the inequality $\frac{\alpha\mathrm{Q}(\hat{\theta}_\mathrm{W},\mathbf{W})+\sigma_\mathrm{W}^2}{\alpha\mathrm{Q}(\hat{\theta}_\mathrm{W},\mathbf{W}_\mathrm{C})+\sigma_\mathrm{W}^2}  \leq  x_2(\epsilon)$ can be rewritten as \eqref{cc}, where
\setcounter{equation}{51}
\begin{equation}
\overline{\mathbf{W}} = \big({(x_2(\epsilon)-1)\mathbf{W}_\mathrm{C} -\mathbf{W}_\mathrm{B}}\big).	
\end{equation} 
Then, based on \textbf{Theorem \ref{Theorem4}}, the covertness constraint~\eqref{c_im} is equivalent to finding a scalar $\zeta \geq 0$ such that the following LMI condition is satisfied,
\begin{equation}\label{co}
	\mathbf{J}  +
	\begin{bmatrix}
		\zeta\frac{\mathbf{I^{(N_t \times N_t)}}}{K \cdot \mathrm{CRB}(\hat{\theta}_\mathrm{W})}& \mathbf{0}\\
		{\mathbf{0}} &  - \zeta 
	\end{bmatrix} 
	\succeq 0,	
\end{equation}
where 
\begin{equation*}
	\small
	\mathbf{J} = \begin{bmatrix}
		\overline{\mathbf{W}}  & \overline{\mathbf{W}}{\mathbf{a}}_A(\hat{\theta}_\mathrm{W})\\
		{\mathbf{a}}_A(\hat{\theta}_\mathrm{W})^\mathrm{H}\overline{\mathbf{W}}  & {\mathbf{a}}_A(\hat{\theta}_\mathrm{W})^\mathrm{H}\overline{\mathbf{W}}{\mathbf{a}}_A(\hat{\theta}_\mathrm{W}) + \frac{d_{\mathrm{AW}}^{\alpha_\mathrm{W}}\sigma_\mathrm{W}^2(x_2(\epsilon)-1)}{\beta_\mathrm{W}}
	\end{bmatrix}.
\end{equation*}
As a result, the original infinite-dimensional constraint~\eqref{c_im} is transformed into the finite-dimensional constraint. However, the constraint \eqref{co} is still difficult to tackle due to the involved $\mathrm{CRB}(\hat{\theta}_\mathrm{W})$. To this end, we further introduce a auxiliary variable $t$ which satisfies $t \geq t_0$, and a tighter version of \eqref{sensing} is expressed as
\begin{equation}\label{sensing1}
	\mathbf{U} - \begin{bmatrix}
		t & 0 \\
		0 & 0
	\end{bmatrix} \succeq 0. 	
\end{equation}
Moreover, since there exists $\frac{\sigma_\mathrm{A}^2}{{2|\alpha|^2L}t} \leq \Gamma_\mathrm{CRB}$, the $\mathrm{CRB}(\hat{\theta}_\mathrm{W})$ can be equivalently replaced by $\frac{\sigma_\mathrm{A}^2}{{2|\alpha|^2L}t}$, and the constraint \eqref{co} is reformulated as  
\begin{equation}\label{co1}
	\mathbf{J}  +
	\begin{bmatrix}
		\zeta\frac{{{2|\alpha|^2L}t}}{K \cdot {\sigma_\mathrm{A}^2}}\mathbf{I^{(N_t \times N_t)}}& \mathbf{0}\\
		{\mathbf{0}} &  - \zeta 
	\end{bmatrix} 
	\succeq 0.	
\end{equation}
 
Finally, The optimization problem $\mathcal{P}2$ is transformed into
\begin{subequations}\label{P2.1}
	\begin{align}
		(\mathcal{P}2.1)~&\mathop {{\rm{Max}}}\limits_{\mathbf{W}_\mathrm{B},\mathbf{W}_\mathrm{C},\mathbf{V}, t, \zeta} 
		\overline{\operatorname{Tr}(\mathbf{V}\mathbf{G}_\mathrm{B}\mathbf{W}_\mathrm{B}\mathbf{G}_\mathrm{W}^\mathrm{H})} - \eta \big(\big\|\mathbf{V}\big\|_* - \overline{\mathbf{V}}^{(n)} \big),\label{4a} \\
		\text{s.t.} ~~~~
		& t \geq t_0, \label{4b} \\
		& \zeta \geq 0, \label{4c} \\
		& \eqref{c1}, \eqref{c2}, \eqref{2d},\eqref{2e},\eqref{2g}, \eqref{2h}, \eqref{sensing1}, \eqref{co1}. \label{4d}
	\end{align}
\end{subequations}
The only challenge in solving $\mathcal{P}2.1$ arises from the coupling of two auxiliary variables in constraint \eqref{co1}, for which we adopt an alternating optimization (AO) based approach. Specifically, the problem $\mathcal{P}2.1$ is decomposed into two subproblems with respect to $\big\{ {\mathbf{W}_\mathrm{B},\mathbf{W}_\mathrm{C},\mathbf{V}, t} \big\}$ and ${\zeta}$, respectively.
Given $\zeta^{(k)}$ at the $k$-th iteration, the first subproblem is given as
\begin{subequations}\label{P2.1a}
	\begin{align}
		(\mathcal{P}2.1a)~~&\mathop {{\rm{Max}}}\limits_{\mathbf{W}_\mathrm{B},\mathbf{W}_\mathrm{C},\mathbf{V}, t} 
		\overline{\operatorname{Tr}(\mathbf{V}\mathbf{G}_\mathrm{B}\mathbf{W}_\mathrm{B}\mathbf{G}_\mathrm{B}^\mathrm{H})} - \eta \big(\big\|\mathbf{V}\big\|_* - \overline{\mathbf{V}}^{(n)} \big),\label{4a} \\
		~~\text{s.t.} ~~~~
		&\eqref{4b}, \eqref{4d}.
	\end{align}
\end{subequations}
 Given $\big\{ {\mathbf{W}_\mathrm{B}^{(k)}, \mathbf{W}_\mathrm{C}^{(k)}, t^{(k)}} \big\}$, the second subproblem is formulated as
\begin{subequations}\label{P2.1b}
	\begin{align}
		(\mathcal{P}2.1b)~~\quad & \text{Find}~ \zeta, \label{4a} \\
		\text{s.t.} \quad & \eqref{co1},~\eqref{4c}.
	\end{align}
\end{subequations}
Since $\mathcal{P}2.1a$ is a concave SDP problem and $\mathcal{P}2.1b$ is a linear feasibility problem, both of them can be directly solved using the CVX toolbox.

\begin{algorithm}[tbp]
	\caption{Joint optimization algorithm for~$\mathcal{P}2$}
	\label{algorithm2}
	\begin{algorithmic}[1] 
		\STATE Initialize feasible points $\big\{\mathbf{V}^{(0)}, \mathbf{W}_\mathrm{B}^{(0)}, \mathbf{W}_\mathrm{C}^{(0)}\big\}$, and the auxiliary variable $\zeta^{(0)}$.
		\STATE Set the iteration index $k=0$ for the overall AO loop.
		\STATE \textbf{Repeat}
		\STATE \hspace{8pt} Initialize the penalty factor $\eta^{(0)}$.
		\STATE \hspace{8pt} Set $i=0$ for outer loop of the subproblem~\eqref{P2.1a}.
		\STATE \hspace{8pt} \textbf{Repeat}
		\STATE \hspace{20pt} Set $n=0$ for inner loop of the subproblem~\eqref{P2.1a}.
		\STATE \hspace{20pt} \textbf{Repeat}
		\STATE \hspace{32pt} Solve~\eqref{P2.1a} with $\big\{\mathbf{V}^{(n)}, \mathbf{W}_\mathrm{B}^{(n)}, \mathbf{W}_\mathrm{C}^{(n)}\big\}$, $\eta^{(i)}$, $\zeta^{(k)}$.
		\STATE \hspace{32pt} Update $\big\{\mathbf{V}^{(n+1)}, \mathbf{W}_\mathrm{B}^{(n+1)}, \mathbf{W}_\mathrm{C}^{(n+1)},t^{(n+1)}\big\}$.
		\STATE \hspace{32pt} Update $n=n+1$.
		\STATE \hspace{21pt} \textbf{Until} the predefined accuracy or maximum iteration number is reached.
		\STATE \hspace{21pt} Update  $\big\{\mathbf{V}^{(0)}, \mathbf{W}_\mathrm{B}^{(0)}, \mathbf{W}_\mathrm{C}^{(0)}\big\}$=$\big\{\mathbf{V}^{(n)}, \mathbf{W}_\mathrm{B}^{(n)}, \mathbf{W}_\mathrm{C}^{(n)}\big\}$.
		\STATE \hspace{21pt} Update $\eta^{(i+1)} = \omega \eta^{(i)}$, and $i=i+1$.
		\STATE \hspace{8pt} \textbf{Until} $ \big\|\mathbf{V}\big\|_*^{(i)} - \big\|\mathbf{V}\big\|_2^{(i)} \leq \mu$.
		\STATE \hspace{8pt} $\big\{\mathbf{V}_\mathrm{B}^{(k)},\mathbf{W}_\mathrm{B}^{(k)}, \mathbf{W}_\mathrm{C}^{(k)}, t^{(k)}\big\}$ = $\big\{\mathbf{V}_\mathrm{B}^{(0)},\mathbf{W}_\mathrm{B}^{(0)}, \mathbf{W}_\mathrm{C}^{(0)}, t^{(0)}\big\}$.
		\STATE \hspace{8pt} Solve~\eqref{P2.1b} with $\big\{\mathbf{W}_\mathrm{B}^{(k)}, \mathbf{W}_\mathrm{C}^{(k)}, t^{(k)}\big\}$.
		\STATE \hspace{8pt} Update $\zeta^{(k+1)}$, and $k=k+1$.
		\STATE \textbf{Until} the predefined accuracy or maximum iteration number is reached.
	\end{algorithmic}
\end{algorithm}

The overall algorithm for the proposed optimization problem is concluded as \textbf{Algorithm~\ref{algorithm2}}. Note that the proposed algorithm consists of three nested loops. In the outermost iteration, the two subproblems~\eqref{P2.1a} and~\eqref{P2.1b} are solved alternately. Similar to \textbf{Algorithm~\ref{algorithm1}}, the solution of subproblem~\eqref{P2.1a} involves two loops based on the penalty method. Since subproblem~\eqref{P2.1b} is a feasibility problem with respect to the auxiliary variable~$\zeta$, it does not affect the value of the objective function. Moreover, the objective function of subproblem~\eqref{P2.1a} increases monotonically and is upper-bounded during the iterations. Therefore, \textbf{Algorithm~\ref{algorithm2}} is guaranteed to converge. Ignoring the computational complexity of solving~\eqref{P2.1b}, the overall complexity of the AO-based algorithm is given as~$\mathcal{O}(I_A I_o I_i (N_t^{3.5} + M^{3.5}))$, where $I_A$ denotes the number of iterations required for AO convergence.

\section{Simulation Results}\label{simulation}
In this section, numerical results are provided to compare the system performance under perfect and imperfect CSI models at Willie, to demonstrate the advantages of NOMA over OMA, and to validate the dual effect of sensing discussed in \textbf{Remark~\ref{remark1}}. Unless otherwise specified, the simulation settings are given as \textbf{Table~\ref{table}}.

During the development of the sensing-based imperfect CSI model, \eqref{delta_a} is approximated by \eqref{appro} for tractability. The feasibility of this approximation is demonstrated in Fig.\ref{Fig2}. As observed, when the angle error $\Delta\theta_W$ is within $[-2^\circ,2^\circ]$, the difference between the values of \eqref{delta_a} and \eqref{appro} is very small and thus negligible. For the proposed optimization problem~\eqref{P2}, the sensing accuracy constraint in~\eqref{3a} ensures that the achieved angle error remains at a sufficiently low level. Therefore, the proposed sensing-based imperfect CSI model can be regarded as reasonably accurate.

In addition to the perfect CSI in the known-location case and the sensing-based imperfect CSI model in the unknown-location case, we further consider the general norm-bounded imperfect CSI model, where the uncertainty is limited by a predefined constant, as a benchmark. Moreover, to evaluate the advantage of NOMA, the conventional OMA transmission scheme is included for comparison.
\begin{itemize}
	\item \textbf{Norm-bounded imperfect CSI model:} In this model, the CSI uncertainty is bounded by a predefined constant. Note that in the sensing-based imperfect CSI model, the maximum tolerable CRB is set as $\Gamma_\mathrm{CRB} = 4\times 10^{-6}$, which corresponds to a maximum angle error of $3\sqrt{\Gamma_\mathrm{CRB}}$, i.e., $0.344^\circ$. Therefore, for fair comparison, the norm-bounded imperfect CSI is given as
	\begin{equation*}
			\overline{\mathbf{h}}_{\mathrm{AW}}({\theta}_\mathrm{W} ) =\sqrt{\frac{\beta_\mathrm{W}}{d_{\mathrm{AW}}^{\alpha_\mathrm{W}}}}  \mathbf{a}_\mathrm{A}(\hat{\theta}_\mathrm{W} +\Delta\theta_\mathrm{W}), |\Delta\theta_\mathrm{W}| \leq 0.344^\circ.
	\end{equation*}
    \item \textbf{OMA transmission:} In this scheme, Bob and Carol occupy orthogonal time resources and communicate with Alice in sequence. The corresponding SINRs are given by
    $\mathrm{SINR}_\mathrm{B}^{\mathrm{OMA}} = \frac{| \mathbf{g}_\mathrm{B}\mathbf{w}_\mathrm{B}|^2}{\sigma_\mathrm{B}^2}, $ and
    $\mathrm{SINR}_\mathrm{C}^{\mathrm{OMA}}  = \frac{| \mathbf{g}_\mathrm{C} \mathbf{w}_\mathrm{C}|^2}{ \sigma_\mathrm{C}^2}.$
    Accordingly, the achievable transmission rates for Bob and Carol are expressed as
    $R_\mathrm{B}^{\mathrm{OMA}} = \tau_\mathrm{B} \log_2(1+\mathrm{SINR}_\mathrm{B}^{\mathrm{OMA}}),$ and 
    $R_\mathrm{C}^{\mathrm{OMA}} = \tau_\mathrm{C} \log_2(1+\mathrm{SINR}_\mathrm{C}^{\mathrm{OMA}}),$ respectively,
    where $\tau_\mathrm{B}$ and $\tau_\mathrm{C}$ denote the time allocation coefficients for Bob and Carol, respectively, satisfying $\tau_\mathrm{B} + \tau_\mathrm{C} = 1$. 
    The variables $\{\tau_\mathrm{B},\tau_\mathrm{C}\}$ are optimized as a separate subproblem via linear programming. 
\end{itemize}

\begin{table}[tbp]
	\centering
	\small
	\caption{Simulation Parameters}
	\label{table}
	\begin{tabular}{ll}
		\toprule
		\textbf{Parameter}   & \textbf{Value} \\
		\midrule
		Carrier frequency  & $f_c = 5$~GHz   \\
		Wavelength         & $\lambda = 0.06$~m   \\
		Element spacing    & $ d = 0.03$~m\\
		Noise power        & $\sigma_B^2,\sigma_C^2,\sigma_A^2,\sigma_W^2 = -70$~dBm  \\
		Fading coefficient & {\footnotesize $\beta_R,\beta_B,\beta_C = -30$~dB, $\beta_W=-20$~dB}  \\
		Path loss exponent & $\alpha_R,\alpha_B,\alpha_C = 2.4$, $\alpha_W =2$      \\
		Rician factor      & $\ell = 10$ \\
		Communication distance    & {\footnotesize $d_\mathrm{AR} = 50$~m, $d_\mathrm{RB} = 5$~m, $d_\mathrm{RC} = 8$~m} \\
		Sensing distance    &$d_\mathrm{AW} = 5$~m \\
		AoD of Alice       & $\theta_R=3^\circ$, $\hat{\theta}_W= 0^\circ$ \\
		Elevation angle of RIS & $\gamma_B,\gamma_C=0^\circ$\\ 
		Azimuth angle of RIS   & $\phi_B = 60^\circ$, $\phi_C = 150^\circ$ \\
		Transmitted antennas   & $N_t$ = 10 \\ 
		Received Antennas      & $N_r$ = 12  \\
		Element number         & $M =5\times 5$ \\
		Channel use/ snapshot & $L =30$ \\
		Power budget           & $P$ = 30~dBm  \\
		QoS requirement        & $R_0 = 1$~bps/Hz \\
		Covertness level       & $\epsilon = 0.1$        \\
		Sensing accuracy       & $\Gamma_\mathrm{CRB} = 4\mathrm{e}^{-6}$     \\
		Initialized penalty factor        & $\eta = 1\mathrm{e}^{-4}$    \\
		Scaling factor         & $\omega = 10$ \\
		\bottomrule
	\end{tabular}
\end{table}

\begin{figure}[tbp]
	\centering
	\includegraphics[width=2.7in]{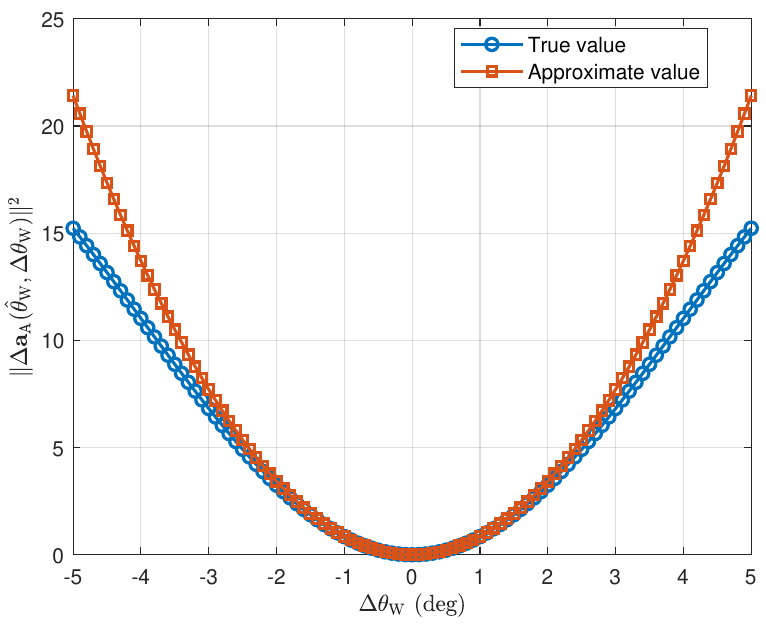}\\
	\caption{Accuracy of the sensing-based imperfect CSI model.}\label{Fig2}
\end{figure}

Combining the three channel models with two multiple access strategies results in six schemes\footnote{For brevity, perfect CSI, sensing-based imperfect CSI, and norm-bounded imperfect CSI are referred to as PC," SBIC," and ``NBIC," respectively, in the simulation figures.}, and we evaluate the performance of these schemes under different system settings in Fig.~\ref{Fig3} and Fig.~\ref{Fig4}. 

The trade-off between Carol's public communication and Bob's covert communication is investigated in Fig.~\ref{Fig3}. First, it is observed that Bob’s covert rate decreases as the QoS threshold $R_0$ of Carol increases under all six schemes. This degradation results from the competition for limited communication resources between Bob and Carol. Specifically, in the case of NOMA, the competing resources include the RIS beamforming and Alice’s power allocation, while in the case of OMA, the competition additionally involves the allocation of orthogonal time resources. A higher $R_0$ necessitates that more resources be allocated to Carol, which consequently reduces the resources available for Bob and ultimately leads to a lower covert rate.
Second, consistent with expectations, the perfect CSI model achieves the best performance under both NOMA and OMA schemes, serving as an upper performance bound for the other two channel models. The sensing-based imperfect CSI model significantly outperforms the norm-bounded imperfect CSI model due to the fact that the CRB-bounded CSI uncertainty in the former can be smaller than the fixed uncertainty assumed in the latter. This tighter uncertainty bound makes robust covert transmission more achievable. As $R_0$ increases, the performance gap among the three channel models gradually narrows. When $R_0 = 3$~bps/Hz, the covert rate achieved under the sensing-based imperfect CSI reaches up to 98.7\% and 93.7\% of that under the perfect CSI model for NOMA and OMA, respectively.
Furthermore, regardless of the adopted channel model, NOMA consistently outperforms OMA when $R_0 \neq 0$. As $R_0$ increases, the performance of OMA schemes degrades rapidly, whereas the decline under NOMA is significantly more gradual. This highlights the advantages of NOMA in enhancing user fairness and facilitating covert communication in multi-user systems.

 \begin{figure}[tbp]
	\centering
	\includegraphics[width=2.7in]{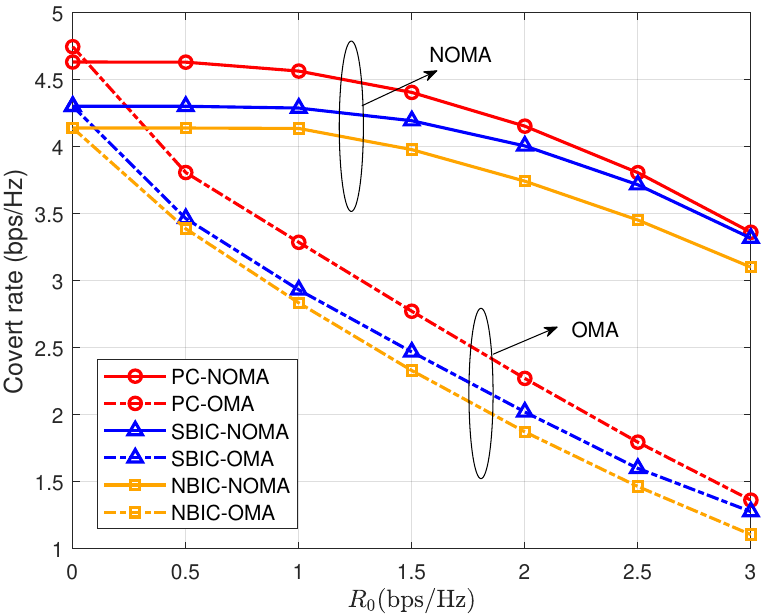}\\
	\caption{Bob-Carol rate trade-off.}\label{Fig3}
\end{figure}
\begin{figure}[tbp]
	\centering
	\includegraphics[width=2.7in]{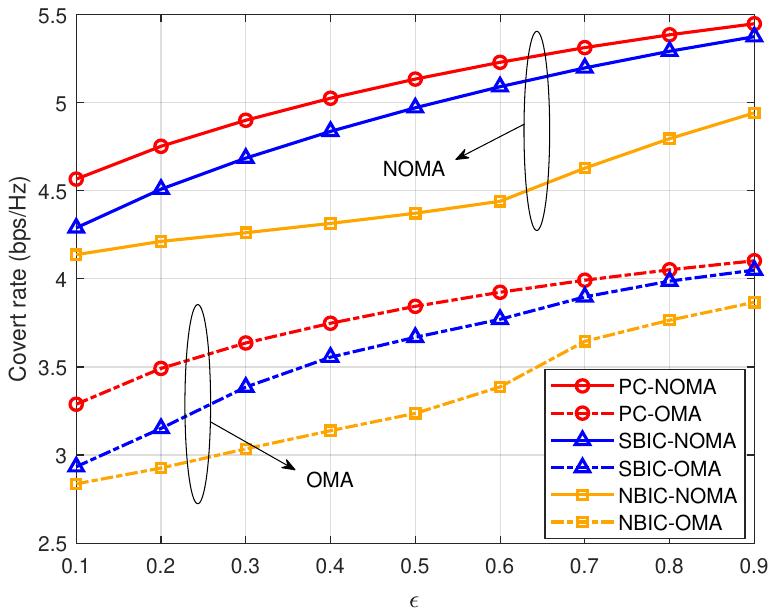}\\
	\caption{Covert rate versus the covertness requirement level.}\label{Fig4}
\end{figure}
In Fig.~\ref{Fig4}, the impact of the covertness requirement level,~$\epsilon$, on system performance is illustrated. It is observed that the covert rate increases as $\epsilon$ becomes larger. This trend can be interpreted in two ways. From the perspective of the covertness constraints, i.e., constraints \eqref{2j} and \eqref{c_im}, a larger~$\epsilon$ corresponds to a larger~$x_2(\epsilon)$, which relaxes the beamforming restrictions at Alice and allows more energy to be allocated to Bob, thereby improving the covert rate. Alternatively, from the viewpoint of the MDEP at Willie, a larger~$\epsilon$ means a lower detection requirement, which naturally enables a higher achievable covert rate.
Moreover, aligned with the findings in Fig.~\ref{Fig3}, the perfect CSI model achieves the highest covert rate, followed by the sensing-based and norm-bounded imperfect CSI models. NOMA consistently outperforms OMA, demonstrating notable performance gains.

\begin{figure}[tbp]
	\centering
	\includegraphics[width=2.8in]{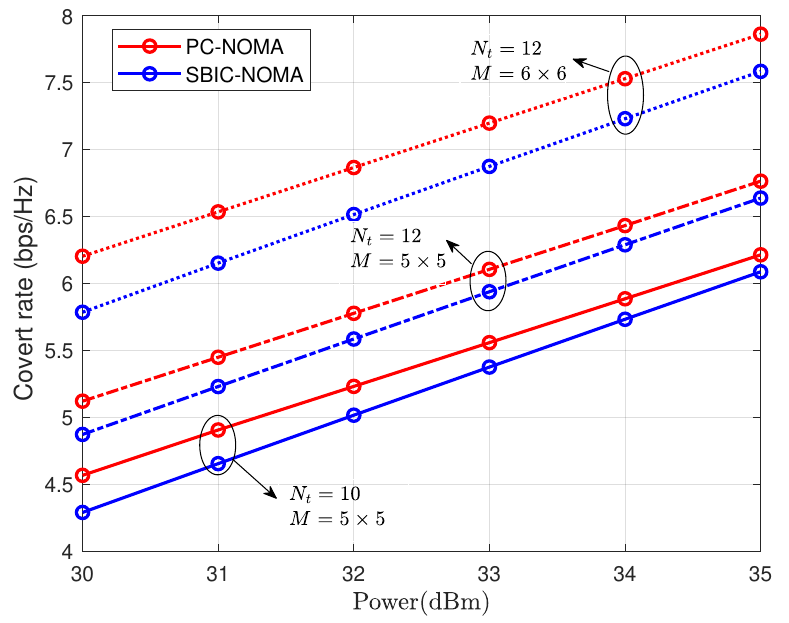}\\
	\caption{Covert rate versus power under different~$M$ and~$N_t$.}\label{Fig5}
\end{figure}

The impact of the power budget on covert communication performance is illustrated in Fig.~\ref{Fig5}, with different configurations of transmit antennas and RIS sizes considered.  
It is shown that the covert rate increases with the maximum transmit power $P$, which is easily expected, as the SINR increases proportionally with transmit power. 
Additionally, it is observed that increasing the number of transmit antennas at Alice improves covert performance. This is because Alice with more transmitted antennas, enables more precise beamforming under the same system settings.
Likewise, increasing the number of RIS element also leads to a higher covert rate, thanks to two factors: more propagation paths and enhanced RIS beamforming capability.

Next, we explore the relationship between sensing accuracy and robust covert communication rate in Fig.~\ref{Fig6}. 
This figure includes dual horizontal axes, which represent the upper bound of sensing accuracy in the sensing-based imperfect CSI model and the upper bound of angle error in the norm-bounded imperfect CSI model, respectively. These two axes correspond to each other according to $|\Delta \theta_\mathrm{W}|_{\max} = 3\sqrt{\Gamma_{\mathrm{CRB}}}$. For the perfect CSI model, the channel uncertainty is zero, and thus the covert transmission rate remains constant. For the norm-bounded imperfect CSI model, its CSI uncertainty increases monotonically with the angle error $|\Delta\theta_\mathrm{W}|$, resulting in a gradual degradation of the robust covert transmission rate. Next, we focus on the sensing-based imperfect CSI model, whose uncertainty is bounded by $K \times \mathrm{CRB}$.
It can be observed that when $\Gamma_\mathrm{CRB} < 3.4\mathrm{e}^{-6}$, both the CSI uncertainty and the covert communication rate increase as $\Gamma_\mathrm{CRB}$ increases. This indicates a trade-off between sensing and communication. Specifically, when $\Gamma_\mathrm{CRB}$ is relatively small, to meet the sensing accuracy constraint \eqref{3b}, more energy is allocated for sensing (i.e., the beams are more aligned toward Willie), which consequently reduces the communication performance.
However, it is noteworthy that when $\Gamma_\mathrm{CRB} \geq 3.4\mathrm{e}^{-6}$, both the CRB and covert transmission rate stop increasing and remain constant. This indicates that constraint \eqref{3b} no longer exerts any influence.
In this case, to achieve robust covert rate maximization under the constraint \eqref{c_im}, an optimal beamforming is designed to maintain the CRB at an appropriate value~(0.0827), as discussed in \textbf{Remark~\ref{remark1}}.
Comparing the norm-bounded imperfect CSI model and the sensing-based imperfect CSI model, it is found that when $\Gamma_{\mathrm{CRB}}$ or the angle error bound is relatively small, the former achieves a higher covert communication rate, owing to the inherent trade-off between sensing and communication. However, as $\Gamma_{\mathrm{CRB}}$ or the angle error increases, the latter gradually demonstrates its performance advantage. The results shown in Fig.6 indicate that, in practical implementations of the proposed joint sensing and covert communication system, the selection of an appropriate value of $\Gamma_{\mathrm{CRB}}$ plays a crucial role in fully exploiting the performance benefits. Under the parameters listed in \textbf{Table~\ref{table}}, it is required that $\Gamma_{\mathrm{CRB}} \geq 3.4 \times 10^{-6}$.
\begin{figure}[tbp]
	\centering
	\includegraphics[width=2.9in]{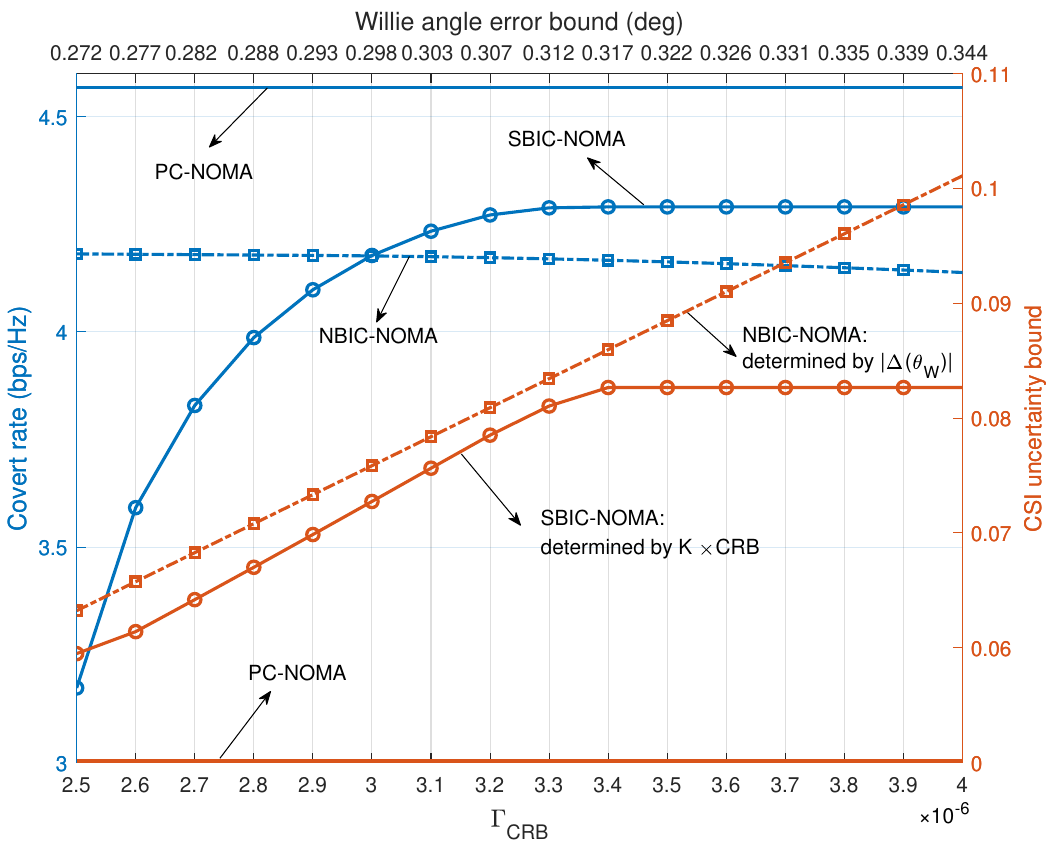}\\
	\caption{Sensing-communication relationship versus $\Gamma_\mathrm{CRB}$.}\label{Fig6}
\end{figure}
\begin{figure}[tbp]
	\centering
	\subfigure[$N_r = 12$]{
		\includegraphics[width=2.8in]{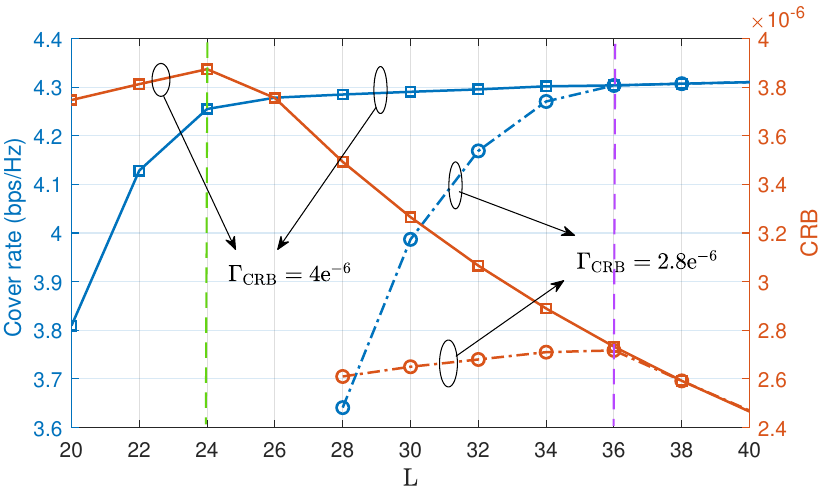}
		\label{a}}
	\subfigure[$N_r = 14$]{
		\includegraphics[width=2.8in]{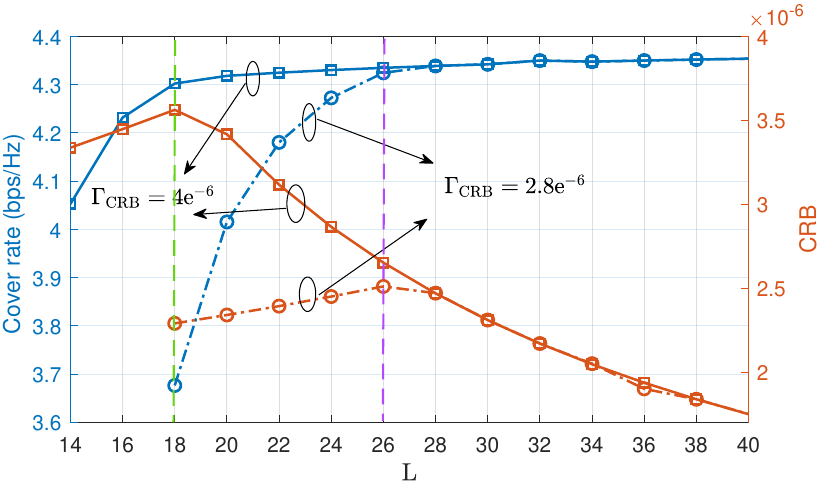}
		\label{b}}		
	\caption{Sensing-communication relationship versus $L$}
	\label{Fig7}
\end{figure}

Fig.~\ref{a} and Fig.~\ref{b} depicts the variation trends of CRB and covert rate with respect to the snapshot $L$,  under $N_r = 12$ and $N_r = 14$, respectively\footnote{The starting points of the horizontal axis for the curves with $\Gamma_\mathrm{CRB}=2.8\times10^{-6}$ and $\Gamma_\mathrm{CRB}=4\times10^{-6}$ are different. This is because, in Fig.\ref{a} and Fig.\ref{b}, the condition $\mathrm{CRB} \leq 2.8\times10^{-6}$ cannot be satisfied when $L<28$ and $L<18$, respectively, rendering the proposed optimization problem infeasible.}.
As observed in Fig.~\ref{a}, when $\Gamma_\mathrm{CRB} = 4\mathrm{e}^{-6}$, both the CRB and the covert rate increase with $L$ for $L < 24$. This is because a smaller number of snapshots corresponds to weaker sensing capability. Under such conditions, the trade-off between sensing accuracy and covert communication rate is the dominant factor.
However, when $L \geq 24$, the CRB decreases while the covert rate keeps increasing with $L$, indicating that enhanced sensing capability facilitates robust covert communication. This is because the constraint $\mathrm{CRB} \leq \Gamma_\mathrm{CRB}$ becomes inactive in this region. As shown in \eqref{CRB}, a larger $L$ leads to a smaller optimized CRB obtained by \textbf{Algorithm~\ref{algorithm2}}, resulting in reduced channel uncertainty and consequently, a higher achievable robust covert rate.
For $\Gamma_\mathrm{CRB} = 2.8\mathrm{e}^{-6}$, a similar trend emerges: sensing accuracy initially hinders the covert rate but later enhances it, with the transition point occurring at $L = 36$.
When $L$ is sufficiently large, the curves under different values of $\Gamma_\mathrm{CRB}$ coincide.

The observations in Fig.~\ref{b} are consistent with those in Fig.~\ref{a}. Under identical parameter settings, Fig.~\ref{b} yields a lower CRB and a higher covert rate compared to Fig.~\ref{a}. Moreover, the transition points appear earlier in Fig.~\ref{b}, with $L = 18$ and $L = 26$ for $\Gamma_\mathrm{CRB} = 4\mathrm{e}^{-6}$ and $\Gamma_\mathrm{CRB} = 2.8\mathrm{e} 10^{-6}$, respectively. This is attributed to the increased number of receive antennas $N_r$, which enhances Alice’s sensing capability.

Furthermore, we study the impact of the angular separation between the RIS and Willie in Fig.~\ref{Fig8}.
Given that~$\theta_\mathrm{W}$ is set as~$0^\circ$, $\theta_\mathrm{R}$ exactly denotes the angular separation.
It can be observed that under the perfect CSI model, the covert rate initially increases with the growth of $\theta_\mathrm{R}$, and then becomes stable with slight fluctuations. 
This is because, with a limited number of antennas at Alice, the spatial resolution of the transmit beam is restricted. A larger angular separation facilitates energy focusing at the RIS, thereby enhancing the communication performance.
For the sensing-based imperfect CSI model, the covert rate initially increases with $\theta_\mathrm{R}$, but then fluctuates and decreases significantly.
The underlying reason is that with this channel model, in addition to the RIS, the beam also needs to be directed toward Willie to obtain its CSI. 
When the angular separation is small, the RIS and Willie share one main beam lobe. However, as the angular separation increases, the beam has to be split into two lobes, resulting in less energy being allocated to the RIS.
Additionally, we observe that with small $\theta_\mathrm{R}$, the achievable performance of the sensing-based imperfect CSI model almost reaches that of the perfect CSI model, regardless of the value of $\Gamma_\mathrm{CRB}$. However, as $\theta_\mathrm{R}$ increases, the covert rate deteriorates with a lower $\Gamma_\mathrm{CRB}$. This is due to the trade-off between the two lobes, one used for communication and the other for sensing.

\begin{figure}[tbp]
	\centering
	\includegraphics[width=2.8in]{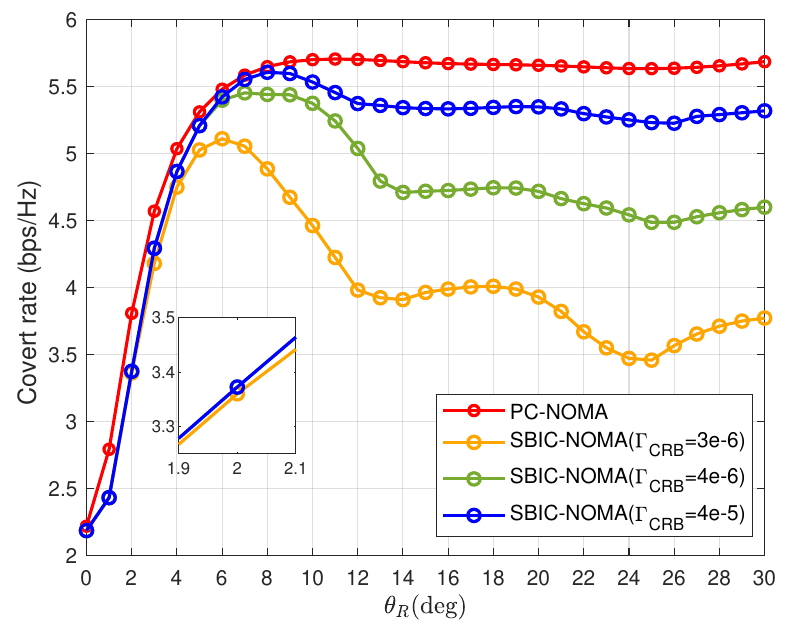}\\
	\caption{The performance under different $\theta_\mathrm{R}$.}\label{Fig8}
\end{figure}
\begin{figure}[tbp]
	\centering
	\includegraphics[width=2.7in]{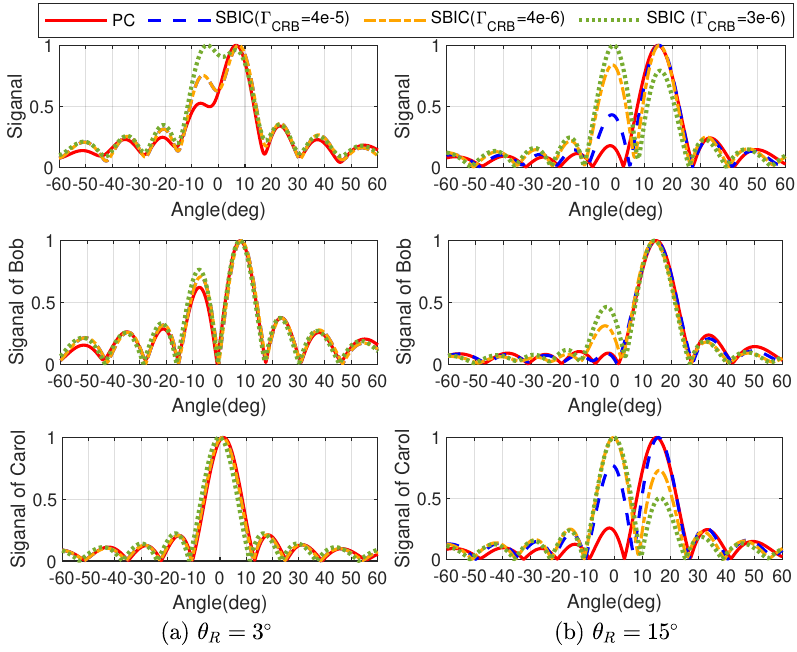}\\
	\caption{Normalized beampattern of Alice.}\label{Fig9}
\end{figure}

Fig.~\ref{Fig9} plots the normalized beampatterns of Alice for both cases of $\theta_\mathrm{R} = 3^\circ$ and $\theta_\mathrm{R} = 15^\circ$, where the superimposed signal, the signal of Bob, and the signal of Carol are shown separately.
It can be seen that when $\theta_\mathrm{R} = 3^\circ$, there exists a single main lobe, whereas for $\theta_\mathrm{R} = 15^\circ$, two lobes are formed. Under the sensing-based imperfect CSI model, a smaller value of $\Gamma_\mathrm{CRB}$ results in a stronger signal directed toward Willie and, correspondingly, a weaker signal received at the RIS. This observation supports the analysis provided for Fig.~\ref{Fig8}.
More importantly, from the beampattern of Bob's signal, it can be observed that its power received at Willie is suppressed to nearly zero, while it reaches the peak  at RIS. This beamforming strategy is favorable for facilitating covert transmission to Bob.
As for Carol's beampattern, the signal is steered toward both the RIS and Willie. In particular, when $\Gamma_\mathrm{CRB}$ is small, a larger portion is directed toward Willie to satisfy the sensing requirement.

\section{Conclusion}\label{conclusion}
This paper investigated the joint sensing and covert communications in a RIS-NOMA system involving a DFRC-BS, a covert user (Bob), a public user (Carol), and a warden (Willie).
For the case of known-location Willie, the active and passive beamforming were jointly optimized under the perfect CSI model, to maximize Bob’s covert rate while satisfying Carol’s QoS requirement.
For the case of unknown-location Willie, a sensing-based imperfect CSI model was proposed with the uncertainty bounded by the CRB, and a robust covert rate optimization problem was formulated accordingly.
Simulation results validated the advantages of the sensing-based imperfect CSI model and NOMA scheme, illustrated the impact of the RIS-Willie angular separation, and examined the implementation of sensing and communication via beampattern analysis.
Notably, it was found that when the sensing accuracy requirement is over-high , sensing tends to constrain communication performance. Otherwise, accurate sensing facilitates more robust covert transmission.

\bibliographystyle{IEEEtran}
\bibliography{mybib}

\begin{IEEEbiography}[{\includegraphics[width=1in, height=1.25in,clip, keepaspectratio]{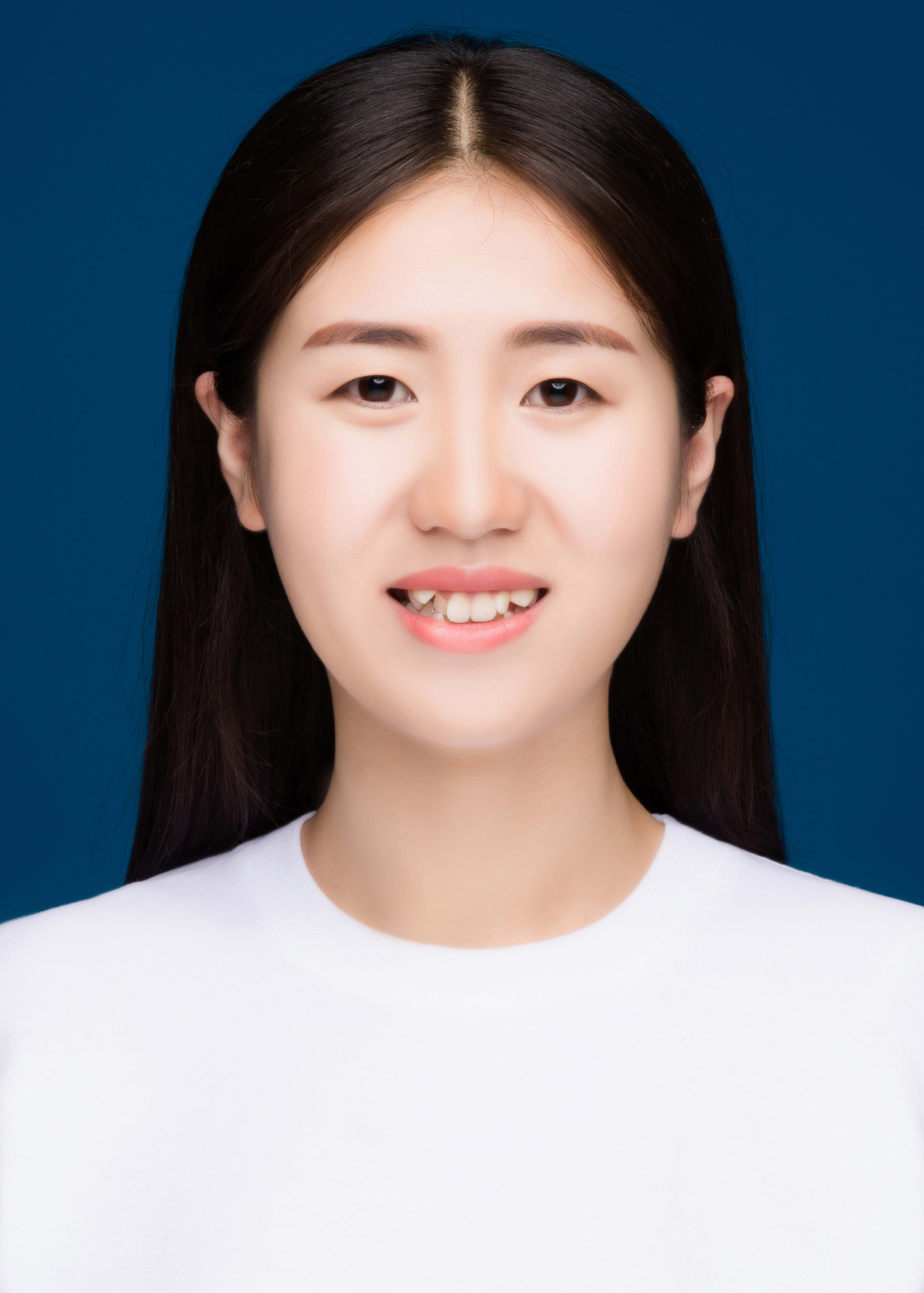}}]
   {Jiayi Lei} (Graduate Student Member, IEEE) received the Ph.D. degree in Information and Communication Engineering from Beijing University of Posts and Telecommunications (BUPT), Beijing, China, in 2025. From 2023 to 2024, she was a Visiting Student at Queen Mary University of London (QMUL), U.K. Her research interests include reconfigurable intelligent surfaces (RIS), non-orthogonal multiple access (NOMA), and integrated sensing and communications (ISAC).
\end{IEEEbiography}

\begin{IEEEbiography}[{\includegraphics[width=1in, height=1.25in,clip, keepaspectratio]{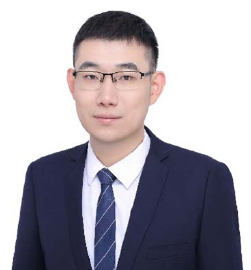}}]
	{Xidong Mu} (Member, IEEE, \url{https://xidongmu.github.io/}) received the Ph.D. degree in Information and Communication Engineering from the Beijing University of Posts and Telecommunications, Beijing, China, in 2022. He was with the School of Electronic Engineering and Computer Science, Queen Mary University of London, from 2022 to 2024, where he was a Postdoctoral Researcher. He has been a lecturer (an assistant professor) with the Centre for Wireless Innovation (CWI), School of Electronics, Electrical Engineering and Computer Science, Queen’s University Belfast, U.K. since August 2024. His research interests include flexible-antenna technologies, reconfigurable surface aided communications, next generation multiple access (NGMA), integrated sensing and communications, and optimization theory. 
	
	Xidong Mu is a Web of Science Highly Cited Researcher. He received the IEEE ComSoc Outstanding Young Researcher Award for EMEA region in 2023 and the IEEE ComSoc Wireless Communications Technical Committee (WTC) Outstanding Young Researcher Award in 2025. He is the recipient of the 2024 IEEE Communications Society Heinrich Hertz Award, the Best Paper Award in ISWCS 2022, the 2022 IEEE SPCC-TC Best Paper Award, and the Best Student Paper Award in IEEE VTC2022-Fall. He serves as the secretary of the IEEE ComSoc Technical Committee on Cognitive Networks (TCCN), the secretary of the IEEE ComSoc NGMA Emerging Technology Initiative, and the URSI UK Early Career Representative (ECR) for Commission C. He also serves as an Editor of \textsc{IEEE Transactions on Communications}, a Guest Editor for \textsc{IEEE Journal on Selected Areas in Communications}, \textsc{IEEE Transactions on Cognitive Communications and Networking}, \textsc{IEEE Internet of Things Journal}, and the “Mobile and Wireless Networks” symposium co-chair of IEEE GLOBECOM 2025. 
\end{IEEEbiography}

\begin{IEEEbiography}[{\includegraphics[width=1in, height=1.25in,clip, keepaspectratio]{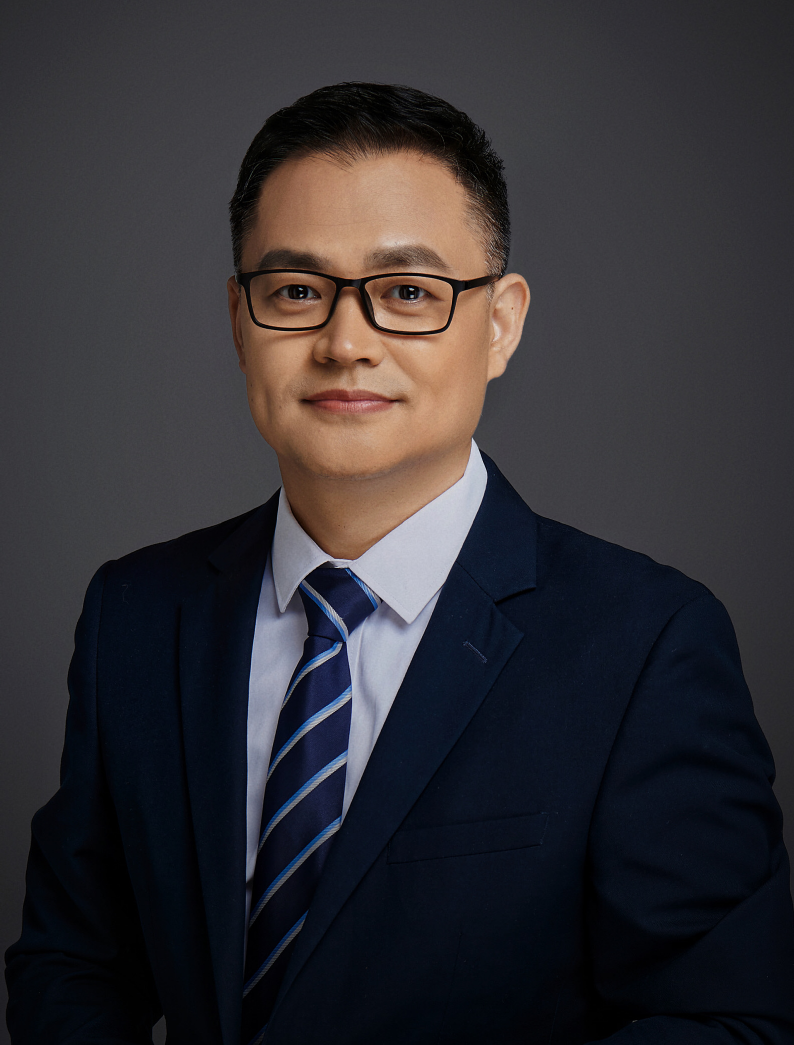}}]
	{Tiankui Zhang} (M'10-SM'15) received the Ph.D. degree in Information and Communication Engineering and B.S. degree in Communication Engineering from Beijing University of Posts and Telecommunications (BUPT), China, in 2008 and 2003, respectively. Currently, he is a Professor in School of Information and Communication Engineering at BUPT. His research interests include artificial intelligence enabling wireless networks, UAV communications in 5G and beyond networks, intelligent mobile edge computing, signal processing for wireless communications. He had published more than 240 papers including journal papers on IEEE Journal on Selected Areas in Communications, IEEE Transaction on Communications, etc., and conference papers, such as IEEE GLOBECOM and IEEE ICC. He is the co-recipient of the Best Paper Award in IEEE GLOBECOM 2022, the Best Paper Award in IEEE APCC 2011. He has served as a TPC Member for many IEEE conferences, such as GLOBECOM and PIMRC, the Technical Program Committee Chair for AiCON 2021. He has served as the guest editor for the special issue “Future Network Architecture and Key Technologies" of the Journal of Beijing University of Posts and Telecommunications, the special issue “Recent Advances in UAV Communications and Networks" of Sensors.
\end{IEEEbiography}

\begin{IEEEbiography}[{\includegraphics[width=1in, height=1.25in,clip, keepaspectratio]{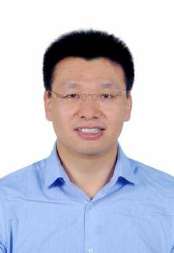}}]
	{Wenjun Xu} (Senior Member, IEEE) received the B.S. and Ph.D. degrees from BUPT, Beijing, China, in 2003 and 2008, respectively. He is currently a Professor and a Ph.D. Supervisor with the School of Artificial Intelligence, State Key Laboratory of Network and Switching Technology, Beijing University of Posts and Telecommunications, Beijing, China. His research interests include AI-driven networks, semantic communications, UAV communications and networks, and green communications and networking. He is serving as an Editor for China Communications.
\end{IEEEbiography}

\begin{IEEEbiography}[{\includegraphics[width=1in, height=1.25in,clip, keepaspectratio]{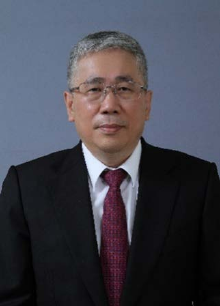}}]
	{Ping Zhang} (M’07-SM’15-F’18) received his Ph.D. degree from Beijing University of Posts and Telecommunications (BUPT) in 1990. He is currently a professor of School of Information and Communication Engineering at BUPT, the director of State Key Laboratory of Networking and Switching Technology, a member of IMT-2020 (5G) Experts Panel, and a member of Experts Panel for China’s 6G development. He served as Chief Scientist of National Basic Research Program (973 Program), an expert in Information Technology Division of National High-tech R\&D program (863 Program), and a member of Consultant Committee on International Cooperation of National Natural Science Foundation of China. He is an Academician of the Chinese Academy of Engineering (CAE). His research is in the board area of wireless communications.
\end{IEEEbiography}

\end{document}